\def\fullversion{yes}
\def\arxivversion{yes}

\ifdefined \arxivversion
\newcommand{\forarxiv}[1]{#1}
\newcommand{\notarxiv}[1]{}
\else
\newcommand{\forarxiv}[1]{}
\newcommand{\notarxiv}[1]{#1}
\fi

\ifdefined \fullversion
\newcommand{\fullversionOnly}[1]{#1}
\newcommand{\confversionOnly}[1]{}
\else
\newcommand{\fullversionOnly}[1]{}
\newcommand{\confversionOnly}[1]{#1}
\fi

\documentclass[10pt]{article}

\usepackage[left=2.5cm, right=2.5cm, top=2.5cm, bottom = 2.5cm]{geometry}
\usepackage{microtype}
\usepackage{graphicx}
\usepackage{subfigure}
\usepackage{mathtools}
\usepackage{natbib}

\usepackage{dsfont}
\usepackage[unicode,colorlinks=true,citecolor=black,filecolor=black,linkcolor=black,urlcolor=black, pdfpagelabels, plainpages=false]{hyperref}
\usepackage{amsfonts,amsmath,amsthm,amssymb,wrapfig}
\usepackage{xcolor}
\usepackage{tikz}
\usetikzlibrary{arrows}
\usepackage{enumitem}
\usepackage{nicefrac}

\newenvironment{algorithm}{}{}

\title{Correlation clustering with the $\ell_q$ objective}
\usetikzlibrary{decorations.pathreplacing}

\newcommand{\ONE}{\mathds{1}}
\newcommand{\br}[1]{\bar{#1}}

\newtheorem{lemma}{Lemma}[section]
\newtheorem{theorem}[lemma]{Theorem}
\newtheorem{claim}[lemma]{Claim}
\newtheorem{corollary}[lemma]{Corollary}

\newtheorem{definition}{Definition}
\numberwithin{figure}{section}

\newcommand{\Eshort}{E_{\leq\varepsilon}}
\newcommand{\Elong}{E_{>\varepsilon}}

\newcommand{\openLP}{\rule{0pt}{12pt}\hrule height 0.8pt\rule{0pt}{1pt}\hrule height 0.4pt\rule{0pt}{1pt}}
\newcommand{\closeLP}{\rule{0pt}{1pt}\hrule height 0.4pt\rule{0pt}{1pt}\hrule height 0.8pt\rule{0pt}{12pt}}

\newcommand{\calP}{\mathcal{P}}
\newcommand{\calC}{\mathcal{C}}
\newcommand{\calS}{\mathcal{S}}
\newcommand{\cP}{\mathcal{P}}
\newcommand{\bE}{\mathbb{E}}
\newcommand{\cE}{\mathcal{E}}
\newcommand{\calE}{\mathcal{E}}
\newcommand{\pr}{\mathrm{Pr}}
\newcommand{\eps}{\varepsilon}

\DeclareMathOperator{\Ball}{Ball}
\DeclareMathOperator{\cut}{cut}
\DeclareMathOperator{\disagree}{disagree}
\DeclareMathOperator{\cost}{cost}
\newcommand{\given}{\;|\;}
\newcommand{\Given}{\;\big|\;}
\DeclareMathOperator*{\argmax}{arg\,max}

\newcommand{\lp}[1]{LP(#1)}
\newcommand{\alg}[1]{ALG(#1)}
\newcommand{\pft}[1]{\text{profit}(#1)}
\newcommand{\prft}[2]{\text{profit}_{#2}(#1)}
\DeclarePairedDelimiter{\roundup}{\lceil}{\rceil}

\begin{document}
\title{Correlation Clustering with Local Objectives}
\author{Sanchit Kalhan \and Konstantin Makarychev \and Timothy Zhou}
\date{Northwestern University}
\maketitle
\begin{abstract}
Correlation Clustering is a powerful graph partitioning model that aims to cluster items based on the notion of similarity between items. An instance of the Correlation Clustering problem consists of a graph $G$ (not necessarily complete) whose edges are labeled by a binary classifier as
``similar'' and ``dissimilar''. An objective which has received a lot of attention in literature is that of minimizing the number of disagreements: an edge is in disagreement if it is a ``similar'' edge and is present across clusters or if it is a ``dissimilar'' edge and is present within a cluster. Define the disagreements vector to be an $n$ dimensional vector indexed by the vertices, where the $v$-th index is the number of disagreements at vertex $v$.
Recently, Puleo and Milenkovic (ICML '16) initiated
the study of the Correlation Clustering framework in which the objectives were more general functions of the disagreements vector. In this paper, we study algorithms for minimizing $\ell_q$ norms $(q \geq 1)$ of the disagreements vector for both arbitrary and complete graphs. We present the first known algorithm for minimizing the $\ell_q$ norm of the disagreements vector on arbitrary graphs and also provide an improved algorithm for minimizing the $\ell_q$ norm $(q \geq 1)$ of the disagreements vector on complete graphs. We also study an alternate cluster-wise local objective introduced by Ahmadi, Khuller and Saha (IPCO '19), which aims to minimize the maximum number of disagreements associated with a cluster. We also present an improved ($2 + \varepsilon$) approximation algorithm for this objective. Finally, we compliment our algorithmic results for minimizing the $\ell_q$ norm of the disagreements vector with some hardness results.

\end{abstract}

\section{Introduction}
A basic task in machine learning is that of clustering items based on the similarity between them. This task can be elegantly captured by Correlation Clustering, a clustering framework first introduced by \citet{BBC04}. In this model, we are given access to items and the \textit{similarity/dissimilarity} between them in the form of a graph $G$ on $n$ vertices. The edges of $G$ represent whether the items are \textit{similar} or \textit{dissimilar} and are labelled as (``$+$'') and (``$-$'') respectively. The goal is to produce a clustering that agrees with the labeling of the edges as much as possible, i.e., to group positive edges in the same cluster and place negative edges across different clusters (a positive edge that is present across clusters or a negative edge that is present within the same cluster is said to be in disagreement). The Correlation Clustering problem can be viewed as an agnostic learning problem, where we are given noisy examples and the task is to fit a hypothesis as best as possible to these examples. Co-reference resolution (see e.g., \citet*{CR01, CR02}), spam detection (see e.g., \citet{RFV07,BGL14}) and image segmentation (see e.g., \citet*{Wirth17}) are some of the applications to which Correlation Clustering has been applied to in practice.

This task is made trivial if the labeling given is consistent (transitive): if $(u,v)$ and $(v,w)$ are similar, then $(u,w)$ is similar for
all vertices $u,v,w$ in $G$ (the connected components on similar edges would give an optimal clustering). Instead, it is assumed that the given labeling is inconsistent, i.e., it is possible that $(u,w)$ are dissimilar even though $(u,v)$ and $(v,w)$ are similar. For such a triplet $u,v,w$, every possible clustering incurs a disagreement on at least one edge and thus, no perfect clustering exists. The optimal clustering is the one which minimizes the disagreements. Moreover, as the number of clusters is not predefined, the optimal clustering can use anywhere from $1$ to $n$ clusters.

Minimizing the total weight of edges in disagreement is the objective that has received the most consideration in literature. Define the disagreements vector be an $n$ dimensional vector indexed by the vertices where the $v$-th coordinate equals the number of disagreements at $v$. Thus, minimizing the total number of disagreements is equivalent to minimizing the $\ell_1$ norm of the disagreements vector. \citet*{PM16} initiated the study of local objectives in the Correlation Clustering framework. They focus on complete graphs and study the minimization of $\ell_q$ norms $(q \geq 1)$ of the disagreements vector -- for which they provided a $48$ approximation algorithm. \citet*{CGS17} gave an improved $7$ approximation algorithm for minimizing $\ell_q$ disagreements on complete graphs. They also studied the problem of minimizing the $\ell_\infty$ norm of the disagreements vector
(also known as Min Max Correlation Clustering) for arbitrary graphs, for which they provided a $O(\sqrt{n})$ approximation.

For higher values of $q$ (particularly $q=\infty$), a clustering optimized for minimizing the $\ell_q$ norm prioritizes reducing the
disagreements at vertices that are worst off. Thus, such metrics are very unforgiving in most cases as it is possible that in
the optimal clustering there is only one vertex with high disagreements while every other vertex has low disagreements. Hence, one is forced to infer the most pessimistic picture about the overall clustering. The $\ell_2$ norm is a solution to this tension between the $\ell_1$ and $\ell_\infty$ objectives. The $\ell_2$ norm of the disagreements vector takes into account the disagreements at each vertex while also penalizing the vertices with high disagreements more heavily. Thus, a clustering optimized for the minimum $\ell_2$ norm gives a more balanced clustering as it takes into consideration both the global and local picture.

Recently, \cite*{AKS} introduced an alternative min max objective for correlation clustering (which we call AKS min max objective).
For a cluster $C \subseteq V$, let us refer to similar edges with exactly one endpoint in $C$ and dissimilar edges with both endpoints in $C$
as edges in disagreements with respect to $C$. We call the weight of all edges in disagreement with $C$ the cost of $C$. Then,
the AKS min max objective asks to find a clustering $C_1,\dots, C_T$ that minimizes the maximum cost $C_i$.
\cite{AKS} give an $O(\log n)$ approximation algorithm for this objective.


\textbf{Our contributions. } In this paper, we provide positive and negative results for Correlation Clustering with the $\ell_q$ objective. We first study the problem of minimizing disagreements on arbitrary graphs. We present the first approximation algorithm minimizing any $\ell_q$ norm $(q \geq 1)$ of the disagreements vector.

\begin{theorem}\label{Main.Thm.}
There exists a polynomial time $O(n^{\frac{1}{2} - \frac{1}{2q}} \cdot \log^{\frac{1}{2} + \frac{1}{2q}} n)$ approximation algorithm for the minimum $\ell_q$ disagreements problem on general weighted graphs.
\end{theorem}

For the $\ell_2$ objective, the above algorithm leads to an approximation ratio of $\tilde{O}(n^{\nicefrac{1}{4}})$, thus providing the first known approximation ratio for optimizing the clustering for this version of the objective. Note that the above algorithm matches the best approximation guarantee of $O(\log n)$ for the classical objective of minimizing the $\ell_1$ norm of the disagreements vector. For the $\ell_\infty$ norm, our algorithm matches the guarantee of the algorithm by \citet*{CGS17} up to $\log$ factors. Fundamental combinatorial optimization problems like \textit{Multicut, Multiway Cut} and \textit{s-t Cut} can be framed as special cases of Correlation Clustering. Thus, Theorem \ref{Main.Thm.} leads
to the first known algorithms for \textit{Multicut, Multiway Cut} and \textit{s-t Cut} with the $\ell_q$ objective when
$q\neq 1$ and $q \neq \infty$. We can also
use the algorithm from Theorem~\ref{Main.Thm.} to obtain $O(n^{\frac{1}{2} - \frac{1}{2q}} \cdot \log^{\frac{1}{2} + \frac{1}{2q}} n)$
bi-criteria approximation for Min $k$-Balanced Partitioning with the $\ell_q$ objective (we omit details here).

Next, we study the case of complete graphs. For this case, we present an improved $5$ approximation algorithm for minimizing any $\ell_q$ norm $(q \geq 1)$ of the disagreements vector.

\begin{theorem}
There exists a polynomial time $5$ approximation algorithm for the minimum $\ell_q$ disagreements problem on complete graphs.
\end{theorem}

We also study the case of complete bipartite graphs where disagreements need to be bounded for only one side of the bipartition, and not the whole vertex set. We give an improved $5$ approximation algorithm for minimizing any $\ell_q$ norm $(q \geq 1)$ of the disagreements vector.

\begin{theorem}\label{thm:complete}
There exists a polynomial time $5$ approximation algorithm for the minimum $\ell_q$ disagreements problem on complete bipartite graphs where disagreements are measured for only one side of the bipartition.
\end{theorem}

In this paper, we also consider the AKS min max objective. For this objective, we give a $(2 + \varepsilon)$ approximation algorithm,
which improves the approximation ratio of $O(\log n)$ given by~\cite*{AKS}.

\begin{theorem}
There exists a polynomial time $(2 + \eps)$ approximation algorithm for the AKS min max problem on arbitrary graphs.
\end{theorem}

\confversionOnly{
Finally, in the full version of this paper (see supplemental materials), we present an integrality gap of $\Omega(n^{\frac{1}{2} - \frac{1}{2q}})$
for minimum $\ell_q$ $s-t$ cut and
prove a hardness of approximation of 2 for minimum $\ell_\infty$ $s-t$ cut.
}

\fullversionOnly{
Our algorithm for the minimum $\ell_q$ disagreements problem is based on rounding the natural convex programming relaxation for this problem. We show that our result is best possible according to this relaxation by providing an almost matching integrality gap. The integrality gap example we provide is for the minimum $\ell_q$ $s-t$ cut problem (a special case of correlation clustering) and show the following result.

\begin{theorem}
The natural convex programming relaxation for the minimum $\ell_q$ disagreements problem has an integrality gap of $\Omega(n^{\frac{1}{2} - \frac{1}{2q}})$ on arbitrary graphs.
\end{theorem}

Finally, we present a hardness of approximation result for minimum $\ell_\infty$ $s-t$ cut.

\begin{theorem}
There is no $\alpha$-approximation algorithm for the min $\ell_\infty$ \textit{s-t cut problem} for $\alpha<2$ unless P = NP.
\end{theorem}
}


\textbf{Previous work.} \citet*{BBC04} showed that it is NP-hard to find a clustering that minimizes the total disagreements, even on complete graphs. They give a constant-factor approximation algorithm to minimize disagreements and a PTAS to maximize agreements on complete graphs. For complete graphs, \citet*{ACN08} presented a randomized algorithm with an approximation guarantee of $3$ to minimize total disagreements. They also gave a $2.5$ approximation algorithm based on LP rounding. This factor was improved to slightly less than $2.06$ by \citet*{CMSY15}. Since, the natural LP is known to have an integrality gap of $2$, the problem of optimizing the classical objective is almost settled with respect to the natural LP. For arbitrary graphs, the best known approximation ratio is $O(\log n)$ (see \citet*{CGW03, DEFI06}).  Assuming the Unique Games Conjecture, there is no constant-factor approximation algorithm for minimizing $\ell_1$ disagreements on arbitrary graphs (see~\citet{CKKRS06}). \citet*{PM16} first studied Correlation Clustering with more local objectives. For minimizing $\ell_q$ $(q \geq 1)$ norms of the disagreements vector on complete graphs, their algorithm achieves an approximation guarantee of $48$. This was improved to $7$ by \citet*{CGS17}. \citet{CGS17} also studied the problem of minimizing the $\ell_\infty$ norm of the disagreements vector on general graphs. They showed that the natural LP/SDP has an integrality gap of $\nicefrac{n}{2}$ for this problem and provided a $O(\sqrt{n})$ approximation algorithm for minimum $\ell_\infty$ disagreements. \citet*{PM16} also initiated the study of minimizing the $\ell_q$ norm of the disagreements vector (for one side of the bipartition) on complete bipartite graphs. The presented a $10$ approximation algorithm for this problem, which was improved to $7$ by \citet*{CGS17}. Recently, \cite{AKS} studied an alternative objective for the correlation clustering problem. Motivated by creating balanced communitites for problems such as image segmentation and community detection in social networks, they propose a new cluster-wise min-max objective. This objective minimizes the maximum weight of edges in disagreement associated with a cluster, where an edge is in disagreement with respect to a cluster if it is a similar edge and has exactly one end point in the cluster or if it is a dissimilar edge and has both its endpoints in the cluster. They gave an $O(\log n)$ approximation algorithm for this objective.

\section{Preliminaries}\label{sec:prelim}

We now formally define the Correlation Clustering with $\ell_q$ objective problem. We will need the following
definition. Consider a set of points $V$ and two disjoint sets of edges on $V$: positive edges $E^+$ and negative edges
$E^-$. We assume that every edge has a weight $w_{uv}$. For every partition $\calP$ of $V$, we say that a positive
edge is in disagreement with $\calP$ if the endpoints $u$ and $v$ belongs to different parts of $\calP$; and
a negative edge is in disagreement with $\calP$ if the endpoints $u$ and $v$ belongs to the same part of $\calP$.
The vector of disagreements, denoted by $\disagree(\calP, E^+, E^-)$, is a $|V|$ dimensional vector
indexed by elements of $V$. Its coordinate $v$ equals
$$
\disagree_u(\calP, E^+, E^-) =  \smashoperator{\sum_{v:(u,v)\in E^+\cup E^-}} w_{uv} \ONE((u,v) \text{ is in disagreement with }\calP).
$$
That is, $\disagree_u(\calP, E^+, E^-)$ is the weight of disagreeing edges incident to $u$. We similarly define a cut vector
for a set of edges $E$:
$$
\cut_u(\calP, E) = \smashoperator{\sum_{v:(u,v)\in E}} w_{uv} \ONE(u \text{ and } v \text{ are separated by }\calP).
$$
We use the standard definition for the $\ell_q$ norm of a vector $x$: $\|x\|_q= (\sum_u x_u^q)^{\frac{1}{q}}$ and
$\|x\|_{\infty}= \max_u x_u$. For a partition $\calP$, we denote by $\calP(u)$ the piece that contains vertex $u$.

\begin{definition}
In the Correlation Clustering problem with $\ell_q$ objective, we are given a graph $G$ on a set $V$ with
two disjoint set of edges $E^+$ and $E^-$ and a set of weights $w_{uv}$. The goal is find a partition $\calP$ that minimizes the $\ell_q$
norm of the disagreements vector, $\|\disagree(\calP, E^+, E^-)\|_q$.
\end{definition}

In our algorithm for Correlation Clustering on arbitrary graphs, we will use a powerful technique of padded
metric space decompositions~(see e.g., \citet*{Bartal96, Rao99, FT03, GKL03}).
\begin{definition}[Padded Decomposition]
Let $(X, d)$ be a metric space on $n$ points, and let $\Delta > 0$. A probabilistic distribution of partitions $\cP$ of $X$ is called a padded decomposition if it satisfies the following properties:
\begin{itemize}
\item Each cluster $C \in \cP$ has diameter at most $\Delta$.
\item For every $u \in X$ and $\varepsilon >  0$,
$\pr(\Ball(u, \delta) \not\subset \cP(u)) \leq D\cdot \frac{\delta}{\Delta}$
where $\Ball(u, \delta) = \{v \in X : d(u,v) \leq \delta\}$
\end{itemize}
\end{definition}

\begin{theorem}[\citet*{FRT03}]\label{prelim:thm:padded-decomposition}
Every metric space $(X,d)$ on $n$ points admits a $D=O(\log n)$ separating padded decomposition. Moreover, there is a polynomial-time algorithm that samples
a partition from this distribution.
\end{theorem} 

\section{Convex Relaxation}\label{sec:lp}
\begin{figure*}
  \centering
\begin{equation*}
\begin{array}{ll@{}llr}
\text{minimize}  & \displaystyle \max\Big(\|y\|_q, \big(\sum\limits_{u \in V} z_u\big)^{\frac{1}{q}}\Big) \tag{P}\\
\text{subject to}
& y_u=\displaystyle\sum_{v:(u,v)\in E^+} w_{uv} x_{uv} + \sum_{v:(u,v)\in E^-} w_{uv} (1 - x_{uv}) \quad &\text{for all } u \in V & \text{(P1)}\\
& z_u=\displaystyle\sum_{v:(u,v)\in E^+} w^q_{uv} x_{uv} + \sum_{v:(u,v)\in E^-} w^q_{uv} (1 - x_{uv}) \quad &\text{for all } u \in V& \text{(P2)}\\ \\
&x_{v_1v_2}+x_{v_2v_3} \geq x_{v_1v_3} \quad &\text{for all } v_1,v_2,v_3 \in V&\text{(P3)}\\
&x_{uv} = x_{vu}\quad &\text{for all } u,v \in V& \text{(P4)}\\
& x_{uv} \in [0,1] &\text{for all } u,v \in V & \text{(P5)}\\
\end{array}
\end{equation*}
  \caption{Convex relaxation for Correlation Clustering with min $\ell_q$ objective for $q < \infty$.}\label{figure:LPRelaxation}
\end{figure*}

In our algorithms for minimizing $\ell_q$ disagreements in arbitrary and complete graphs, we use a convex relaxation given in Figure~\ref{figure:LPRelaxation}. Our convex relaxation for
Correlation Clustering is fairly standard. It is similar to relaxations used in the papers by~\citet*{GVY96, DEFI06, CGW03}. For every pair
of vertices $u$ and $v$, we have a variable $x_{uv}$ that is equal to the distance between $u$ and $v$ in the ``multicut metric''.
Variables $x_{uv}$ satisfy the triangle inequality constraints~(P3). They are also symmetric~(P4) and $x_{uv}\in [0,1]$~(P5). Thus, the set of
vertices $V$ equipped with the distance function $d(u,v)= x_{uv}$ is a metric space.

Additionally, for every vertex $u\in V$, we have variables $y_u$ and $z_u$ (see constraints~(P1) and (P2)) that lower bound the number of disagreeing edges
incident to $u$. The objective of our convex program is to minimize $\max(\|y\|_q, (\sum_{u} z_u)^{\frac{1}{q}})$. Note that all constraints in
the program (P) are linear; however, the objective function of (P) is not convex as is. So in order to find the optimal solution, we raise
the objective function to the power of $q$ and find feasible $x,y,z$ that minimizes the objective $\max(\|y\|^q_q, \sum_{u} z_u)$.

Let us verify that program (P) is a relaxation for Correlation Clustering. Consider an arbitrary partitioning $\calP$ of $V$.
In the integral solution corresponding to $\calP$, we set $x_{uv} = 0$ if $u$ and $v$ are in the same cluster in $\calP$; and $x_{uv} = 1$ if $u$ and $v$ are
in different clusters in $\calP$. In this solution, distances $x_{uv}$ satisfy triangle inequality constraints~(P3) and
$x_{uv} = x_{vu}$ (P4). Observe that a positive edge $(u,v)\in E^+$ is in disagreement with $\calP$ if  $x_{uv} = 1$; and a negative edge
$(u,v)\in E^-$ is in disagreement if  $x_{uv} = 0$. Thus, in this integral solution, $y_u = \disagree_u(\calP, E^+,E^-)$ and
moreover, $z_u \leq y^q_u$. Therefore, in the integral solution corresponding to $\calP$, the objective function of (P) equals
$\|\disagree_u(\calP, E^+,E^-)\|_q$. Of course, the cost of the optimal fractional solution to the problem may be
less than the cost of the optimal integral solution. Thus, (P) is a relaxation for our problem. Below, we
denote the cost of the optimal fraction solution to (P) by $LP$.

We remark that we can get a simpler relaxation by removing variables $z$ and changing the objective function to $\|y\|_q$.
This relaxation also works for $\ell_{\infty}$ norm. We use it in our 5-approximation algorithm. 

\section{Correlation Clustering on Arbitrary Graphs}

In this section, we describe our algorithm for minimizing $\ell_q$ disagreements on arbitrary graphs. We will prove the following main theorem.

\begin{theorem}\label{thm:arbit-Graphs}
There exists a randomized polynomial-time $O(n^{\frac{q-1}{2q}}\log^{\frac{q+1}{2q}} n)$ approximation algorithm for Correlation Clustering with the $\ell_q$ objective ($q\geq 1$).
\end{theorem}

We remark that the same algorithm gives $O(\sqrt{n\log n})$ approximation for the $\ell_{\infty}$ norm. We omit the details in the conference version of the paper.

\noindent Our algorithm relies on a procedure for partitioning arbitrary metric spaces into pieces of small diameter, which we describe first.

\subsection{Algorithm for Partitioning Metric Spaces}\label{sec:partition-metric-spaces}
In this section, we will prove the following main theorem,


\begin{theorem}\label{thm:part-metric-spaces}
There exists a polynomial-time randomized algorithm that given a metric space $(X,d)$ on $n$ points and parameter $\Delta$
returns a random partition $\calP$ of $X$ such that the diameter of every set $P$ in $\cP$ is at most $\Delta$ and for every $q\geq 1$ ($q\neq \infty$) and
every weighted graph $G=(X,E,w)$, we have
\begin{multline}\label{eq:thm:part-metric-spaces}
\bE\Big[\|\cut(\cP, E)\|_q\Big] \leq C n^{\frac{q-1}{2q}}\log^{\frac{q+1}{2q}} n \cdot \Big[
\Big(\sum_{u\in X}\smashoperator[r]{\sum_{v:(u,v) \in E}}  w^q_{uv} \frac{d(u,v)}{\Delta}\Big)^{1/q} + \\+
\Big(\sum_{u\in X}\Big(\sum_{v:(u,v) \in E}  w_{uv} \frac{d(u,v)}{\Delta} \Big)^q\Big)^{1/q} \Big],
\end{multline}
for some absolute constant $C$.
\end{theorem}

We remark that our algorithm also works for $q=\infty$. Indeed, the behaviour of the algorithm does not depend on $q$ (in fact,
$q$ is not even a part of the algorithm's input). Hence, inequality~(\ref{eq:thm:part-metric-spaces}) holds for any $q<\infty$. In the limit
as $q$ tends to infinity, we get the following result.
\confversionOnly{
We provide the details in the full version of the paper (see supplemental materials for details).}
\fullversionOnly{

\begin{corollary}
The following inequality holds for a random partition $\calP$ from Theorem~\ref{thm:part-metric-spaces}:
$$\bE\Big[\|\cut(\cP, E)\|_{\infty}\Big] \leq  C n^{\frac{1}{2}}\log^{\frac{1}{2}} n \cdot \Big[
\max_{(u,v)\in E} w \cdot \ONE (d(u,v)\neq 0)+ \max_{u\in V}\Big(\sum_{v:(u,v) \in E}  w_{uv} \frac{d(u,v)}{\Delta} \Big)\Big].
$$
\end{corollary}
}

We will need the following definition.

\begin{definition}
Let $(X,d)$ be a metric space. The $\varepsilon$-neighborhood of a set $S\subset X$ is the set of points at distance at most $\varepsilon$ from $S$:
$$N_\varepsilon(S) = \{u \in X: \exists v\in S \text{ such that } d(u,v) \leq \varepsilon\}.$$
The $\varepsilon$-neighborhood of the boundary of a partition $\calP$ is the set of points
$$
N_\varepsilon(\partial \cP) = \bigcup_{P\in\cP} (N_{\varepsilon}(P)\setminus P) =  \{ u \in X: \exists v\in X  \text{ s.t. }
d(u,v) \leq \varepsilon \text{ and } \cP(u) \neq \cP(v) \}.
$$
\end{definition}

We first describe an algorithm which succeeds with probability at least $1/2$ and fails with probability at most $1/2$. If the algorithm succeeds it outputs
a random partition $\calP$ of $X$ such that the diameter of every set $P$ in $\cP$ is at most $\Delta$ and for every $q$ and every weighted
graph $G=(X,E,w)$, we have
\begin{multline}\label{eq:cond-bound-on-partition}
\bE\Big[\|\cut(\cP, E)\|_q\given \text{algorithm succeeds}\Big] \leq
C' n^{\frac{q-1}{2q}}\log^{\frac{q+1}{2q}} n \cdot \Big(\sum_{u\in X}
\smashoperator[r]{\sum_{v:(u,v) \in E}}  w^q_{uv} \frac{d(u,v)}{\Delta}\Big)^{1/q} +\\
\Big(\sum_{u\in X}\Big(\sum_{v:(u,v) \in E}  w_{uv} \frac{d(u,v)}{\Delta} \Big)^q\Big)^{1/q}.
\end{multline}
To obtain a valid partition with probability 1, we repeat our algorithm for at most $\roundup{\log_2 n}$ iterations till it succeeds and
output the obtained solution. If the algorithm does not succeed after $\roundup{\log_2 n}$ iterations (which happens with probability at
most $1/n$), we partition the graph using a simple deterministic procedure which we describe in the end of this section.

Our algorithm is based on the procedure for generating bounded padded stochastic decompositions (see Section~\ref{sec:prelim}). First, the algorithm picks a random padded decomposition $\calP$ of the metric space $X$. Then, it finds the $\varepsilon$-neighborhood
 $N_\varepsilon(\partial \cP)$ of the boundary of $\calP$.
Finally, it outputs $\calP$ if $|N_\varepsilon(\partial \cP)| \leq 2 D \varepsilon/\Delta$ and fails otherwise. We present a pseudo-code for our algorithm in
Figure~\ref{fig:Alg1}.
\begin{figure}
\notarxiv{\begin{center}}

\openLP
\smallskip

\begin{algorithm}
\smallskip
\noindent\textbf{Input: } metric space $(X,d)$ and parameter $\Delta > 0$.\notarxiv{\\}

\noindent\textbf{Output: } a random partition $\calP$ of $X$. 

\begin{enumerate}
  \item Let $D=O(\log n)$ be the parameter from Theorem~\ref{prelim:thm:padded-decomposition},
   $\varepsilon = 1/\sqrt{2Dn}$ and $M=2D\varepsilon n/\Delta$.
  \item Draw a random padded decomposition $\cP$ of the metric space $(X, d)$ with parameter $\Delta$ using Theorem~\ref{prelim:thm:padded-decomposition}.
  \item Find the neighborhood $N_\varepsilon(\partial \cP)$ of the partition boundary.
  \item If $|N_\varepsilon(\cP)| \leq M$ then output $\cP$; else fail.
\end{enumerate}
\end{algorithm}
\closeLP
\notarxiv{\end{center}}
\caption{Metric decomposition algorithm.}\label{fig:Alg1}\label{alg:metric-decomposition}
\end{figure}

\subsection{Analysis}

Our algorithm is scale invariant i.e., its output  does not change if we multiply all distances in the metric space $(X,d)$ and the parameter $\Delta$ by some
positive number $\lambda$. Thus, for the sake of analysis, we assume that $\Delta = 1$. Algorithm~\ref{alg:metric-decomposition} succeeds when $N_\varepsilon(\cP)$ has size at most $M$.
Denote this event by $\cE$. We first show that $\pr(\cE) \geq 1/2$.
\begin{lemma}
Algorithm~\ref{alg:metric-decomposition} succeeds with probability at least~$\nicefrac{1}{2}$.
\end{lemma}
\begin{proof}
Let $\bar\cE$ be the complement of the event $\cE$. We need to show that $\pr(\bar\cE)\leq 1/2$. To this end, we bound
the expected size of the set $N_\varepsilon(\cP)$ using the second property of padded decompositions:
\begin{align*}
     \bE[|N_\varepsilon(\partial \cP)|]
     &= \sum_{u \in X} \pr(u \in N_\varepsilon(\partial\cP))\\
     &= \sum_{u \in X} \pr(\Ball(u, \varepsilon) \not\subset \cP(u))\\
     &\leq \sum_{u \in X} D\varepsilon = D\varepsilon n.
\end{align*}
Here, we used that $u \in N_\varepsilon(\partial\cP)$ if and only if $\Ball(u, \varepsilon) \not\subset \cP(u)$.
Now, by Markov's inequality,
$$ \pr(\bar\cE) =  \pr(|N_\varepsilon(\partial\cP)| > \underbrace{2D\varepsilon n}_M) \leq \frac{D\varepsilon n}{2D\varepsilon n} = \frac{1}{2}.$$
\end{proof}

Let $X_{uv}$ be the indicator of the event $\{\calP(u)\neq \calP(v)\}$ i.e., the event that points $u$ and $v$ are separated by the partition $\calP$.
By the second property of padded stochastic decompositions, we have $\bE(X_{uv})=\pr(\calP(u)\neq \calP(v))\leq D\cdot d(u,v)$.
Since $\pr(\calE)\geq 1/2$, for each $(u,v)\in E$, we have
$$\bE[X_{uv}\given \calE]\leq \frac{\bE[X_{uv}]}{\pr(\calE)}\leq 2\bE[X_{uv}] \leq 2D\cdot d(u,v).$$
Consequently,
\begin{align}
\bE[w_{uv}X_{uv}\given \calE]&\leq 2D\cdot w_{uv} d(u,v) \;\;\;\text{ and } \label{eq:ineq-for-EX-uv-A}\\
\bE[w^q_{uv}X^q_{uv}\given \calE] &\leq 2D\cdot w^q_{uv} d(u,v). \label{eq:ineq-for-EX-uv-B}
\end{align}

We split all edges $E$ into two groups: short edges, which we denote by $\Eshort$, and long edges,
which we denote by $\Elong$. Short edges are edges of length at most $\varepsilon$; long edges are edges
of length greater than $\varepsilon$. Note that
$\cut(\cP, E) = \cut(\cP, \Eshort) +  \cut(\cP, \Elong)$.

For every subset $E'\subset E$ (in particular, for $E'= \Eshort$ and $E'=\Elong$), we have
\begin{equation}
\bE\Big[\|\cut(\cP, E')\|^q_q | \calE\Big] =
\sum_{u\in X} \bE\Big[\big(\smashoperator{\sum_{v:(u,v) \in E'}} w_{uv}X_{uv}\big)^q | \calE\Big].
\label{eq:formula-for-Lq}
\end{equation}\
We separately upper bound $\bE[\|\cut(\cP, \Eshort)\|^q_q\given \calE]$ and $\bE[\|\cut(\cP, \Elong)\|^q_q\given \calE]$ using
the formula above and inequalities (\ref{eq:ineq-for-EX-uv-A}), (\ref{eq:ineq-for-EX-uv-B}) and then
use the triangle inequality for $\ell_q$ norms to bound $\bE[\|\cut(\cP, E)\|_q\given \calE]$.

\medskip

\noindent\textbf{Long edges.} Fix a vertex $u$ and consider long edges incident to $u$. Their total weight is upper bounded by
\begin{align*}
\sum_{v:(u,v) \in \Elong} w_{uv}  &\leq \sum_{v:(u,v) \in \Elong} w_{uv}\;\underbrace{\frac{d(u,v)}{\varepsilon}}_{\geq 1}.
\end{align*}
Thus,
\begin{align*}
\Big(\smashoperator[r]{\sum_{v:(u,v) \in \Elong}}  w_{uv}X_{uv}\Big)^q &\leq \Big(\smashoperator[r]{\sum_{v:(u,v) \in \Elong}} w_{uv}\Big)^{q-1}\Big(\smashoperator[r]{\sum_{v:(u,v) \in \Elong}} w_{uv}X_{uv}\Big)\\
&\leq \Big(\smashoperator[r]{\sum_{v:(u,v) \in \Elong}}  \frac{w_{uv} d(u,v)}{\varepsilon} \Big)^{q-1} \Big(\smashoperator[r]{\sum_{v:(u,v) \in \Elong}} w_{uv} X_{uv}\Big).
\end{align*}
Plugging this expression into formula~(\ref{eq:formula-for-Lq}) with $E'=\Elong$ and using inequality~(\ref{eq:ineq-for-EX-uv-A}), we get the following upper bound on
$\bE\Big[\|\cut(\cP, \Elong)\|^q_q \Given \calE\Big]$:
$$
\sum_{u\in X}
\Big(\smashoperator[r]{\sum_{v:(u,v) \in \Elong}}  \frac{w_{uv} d(u,v)}{\varepsilon} \Big)^{q-1} \bE\Big[\smashoperator{\sum_{v:(u,v) \in \Elong}} w_{uv}X_{uv}\given \calE\Big]
 \leq\frac{2D}{\varepsilon^{q-1}} \sum_{u\in X}\Big(\sum_{v:(u,v) \in \Elong}  w_{uv} d(u,v) \Big)^q.
$$

Finally, by Jensen's inequality, we have
\begin{align}
\notag\bE\big[\|\cut(\cP,\Elong)\|_q\given \calE]&= \bE\big[(\|\cut(\cP, \Elong)\|^q_q)^{\frac{1}{q}}\given \calE\big]\\
&\notag\leq \Big(\bE\big[\|\cut(\cP, \Elong)\|^q_q\given \calE\big]\Big)^{\frac{1}{q}}\\
&\leq \Big(\frac{2D}{\varepsilon^{q-1}} \sum_{u\in X}\Big(\sum_{v:(u,v) \in \Elong}  w_{uv} d(u,v) \Big)^q\Big)^{\frac{1}{q}}.\label{eq:ineq-long}
\end{align}

\medskip
\noindent\textbf{Short edges.} To bound $||\cut(\cP, E_{short})||_q $, we will make use of the following lemma.
\begin{lemma}\label{lem:jensen-bound}
Consider non-negative (dependent) random variables $X_1,\dots, X_n$. Suppose that at most $M$ of them are non-zero with probability 1. Then, for every $q\geq 1$,
the following bound holds:
$$\bE\big[(X_1+\cdots+X_n)^q\big]\leq  M^{q-1}\sum_{i=1}^n \bE\big[X_i^q\big].$$
\end{lemma}
\begin{proof}
Let $x_{i_1}, \ldots, x_{i_m}$ be the non-zero random variables in a certain sampling of $X_1,\dots, X_n$ for some $m\leq M$. Suppose that $m\neq 0$.
Using Jensen's inequality, we have
$$\bigg( \frac{x_{i_1} +  \ldots + x_{i_m}}{m} \bigg)^q \leq \frac{1}{m} \sum_{j = 1}^m x_{i_j}^q,$$
and, therefore,
$$\bigg( {x_{i_1} +  \ldots + x_{i_m}} \bigg)^q \leq m^{q-1} \sum_{j = 1}^m x_{i_j}^q\leq M^{q-1} \sum_{j = 1}^m x_{i_j}^q.$$
The inequality above also holds when $m=0$. Thus, the expectation of the left hand side is upper bounded by the expectation of the right hand side.
This concludes the proof.
\end{proof}

Fix a vertex $u$. Observe that if $(u,v)$ is a short edge which is cut by $\calP$ then $v$ must belong to $N_\varepsilon(\partial\cP)$. Thus,
the number of non-zero random variables $X_{uv}$ for a given $u$ and $(u,v)\in\Eshort$ is upper bounded by $|N_\varepsilon(\partial\cP)|$.
If the algorithm succeeds, then $|N_\varepsilon(\partial\cP)| \leq M$. Thus, by Lemma~\ref{lem:jensen-bound},
$$\bE\big[\big(\smashoperator{\sum_{v:(u,v) \in \Eshort}}  w_{uv}X_{uv}\big)^q\given \calE \big]
\leq M^{q-1}\smashoperator{\sum_{v:(u,v) \in \Eshort}}  \bE\big[w_{uv}^qX_{uv}^q\given \calE \big].$$
Plugging this bound into formula~(\ref{eq:formula-for-Lq}) with $E'=\Eshort$ and using inequality~(\ref{eq:ineq-for-EX-uv-B}), we get
the following upper bound on $\bE\Big[\|\cut(\cP, \Eshort)\|^q_q \Given \calE\Big]$:
$$
\sum_{u\in X}\Big(M^{q-1}\smashoperator{\sum_{v:(u,v) \in \Eshort}}  \bE\big[w_{uv}^qX_{uv}^q\given \calE \big] \Big)\leq
2D\,M^{q-1} \sum_{u\in X}\sum_{v:(u,v) \in \Eshort}  w^q_{uv} d(u,v).
$$
Finally, by Jensen's inequality, we have
\begin{equation}\label{eq:ineq-short}
\bE[\|\cut(\cP, \Eshort)\|_q\given \calE] \leq \Big(2D\,M^{q-1}\Big)^{1/q} \Big(\sum_{u\in X}\sum_{v:(u,v) \in \Eshort}  w^q_{uv} d(u,v)\Big)^{1/q}.
\end{equation}

\medskip
To obtain the desired bound~(\ref{eq:cond-bound-on-partition}), we substitute $D=O(\log n)$, $\varepsilon = 1/\sqrt{2Dn}$, and $M=2D\varepsilon n/\Delta$ in bounds~(\ref{eq:ineq-long}) and~(\ref{eq:ineq-short}) and then apply the triangle inequality for the $\ell_q$ norm.

\confversionOnly{
To finish the proof of Theorem~\ref{thm:part-metric-spaces}, we need to describe what we do in the unlikely event that Algorithm~\ref{alg:metric-decomposition} fails
$\roundup{\log_2 n}$ times. In this case, we create a new graph on $X$ with edges between pairs of vertices at distance at most $1/n$ from each other and partition it
into connected components. We analyze this algorithm in the full version of the paper (see supplemental materials for details).}

\fullversionOnly{
To finish the proof of Theorem~\ref{thm:part-metric-spaces}, we describe what we do in the unlikely event that Algorithm~\ref{alg:metric-decomposition} fails
$\roundup{\log_2 n}$ times.
\begin{lemma}\label{lem:simple-n-approx-alg}
There exists a polynomial-time deterministic algorithm that given a metric space $(X,d)$ on $n$ points and parameter $\Delta$
returns a partition $\calP$ of $X$ such that the diameter of every set $P$ in $\cP$ is at most $\Delta$ and
for every $q$ and every weighted graph $G=(X,E,w)$, we have
$$\|\cut(\cP, E)\|_q \leq n \Big(\sum_{u\in X}\Big(\sum_{v:(u,v) \in E}  w_{uv} \frac{d(u,v)}{\Delta} \Big)^q\Big)^{1/q}.$$
\end{lemma}
\begin{proof}
Consider a graph $\tilde G = (X,\tilde E)$ on $X$ with edges $\tilde E = \{(u,v)\in X\times X: d(u,v) \leq \Delta/n\}$. The algorithm partitions
$\tilde G$ into connected components and outputs the result. Note that the diameter of each connected component $P\in \calP$
is less than $\Delta$, since the length of every edge in $\tilde G$ is less than $\Delta/n$. Let $E_{cut}$ be the set of cut edges in graph $G$.
If two vertices $(u,v)$ are separated by $\calP$, then $d(u,v) \geq \Delta/n$. Hence, for every cut edge $(u,v)\in E_{cut}$, we have $n d(u,v)/\Delta \geq 1$.
Thus,
\begin{align*}
\|\cut(\cP, E)\|_q &= \Big(\sum_{u\in X}\Big(\sum_{v:(u,v) \in E}  w_{uv}\Big)^q\Big)^{1/q}\\
&\leq  \Big(\sum_{u\in X}\Big(\sum_{v:(u,v) \in E}  w_{uv} \frac{n\cdot d(u,v)}{\Delta} \Big)^q\Big)^{1/q}\\
&= n  \Big(\sum_{u\in X}\Big(\sum_{v:(u,v) \in E}  w_{uv} \frac{d(u,v)}{\Delta} \Big)^q\Big)^{1/q}.
\end{align*}
\end{proof}
}

\subsection{Proof of Theorem~\ref{thm:arbit-Graphs}}

We now show how to use the above metric space partitioning scheme to obtain an approximation algorithm for Correlation Clustering on arbitrary graphs. 

\begin{proof}[Proof of Theorem~\ref{thm:arbit-Graphs}]
Our algorithm first finds the optimal solution $x,y,z$ to the convex relaxation (P) presented in Section~\ref{sec:lp}. Then, it defines a metric $d(u,v)= x_{uv}$ on
the vertices of the graph. Finally, it runs the metric space partitioning algorithm with $\Delta = 1/2$ from Section~\ref{sec:partition-metric-spaces}
(see Theorem~\ref{thm:part-metric-spaces}) and outputs the obtained partitioning $\calP$.

Let us analyze the performance of this algorithm. Denote the cost of the optimal solution $x,y,z$ by $LP$. We know that the cost of the optimal solution
$OPT$ is lower bounded by $LP$ (see Section~\ref{sec:lp} for details). By Theorem~\ref{thm:part-metric-spaces}, applied to the graph $G=(V,E^+)$
(note: we ignore negative edges for now),
\begin{equation}\label{eq:thm:part-metric-spaces:approx-alg}
\bE\Big[\|\cut(\cP, E^+)\|_q\Big] \leq \frac{C}{\Delta} n^{\frac{q-1}{2q}}\log^{\frac{q+1}{2q}} n \cdot \Big(\big(\sum_{u\in V} y_u^q\big)^{\frac{1}{q}} +
 \big(\sum_{u\in V} z_u\big)^{\frac{1}{q}}\Big)\leq 4C n^{\frac{q-1}{2q}}\log^{\frac{q+1}{2q}} n \cdot LP.
\end{equation}
Recall that a positive edge is not in agreement if and only if it is cut. Hence, $\disagree(\cP,E^+,\varnothing) = \cut(\cP, E^+)$, and the bound above holds
for $\bE \|\disagree(\cP, E^+,\varnothing)\|_q $. By the triangle inequality,
$\bE\|\disagree(\cP, E^+,E^-)\|_q \leq \bE\|\disagree(\cP, E^+,\varnothing)\|_q + \bE\|\disagree(\cP, \varnothing, E^-)\|_q$.
Hence, to finish the proof, it remains to upper bound $\bE\|\disagree(\cP, \varnothing, E^-)\|_q$.

Observe that the diameter of every cluster returned by the algorithm is at most $\Delta = 1/2$. For all disagreeing
negative edges $(u,v)\in E^-$, we have $x_{uv}\leq 1/2$ and $1-x_{uv}\geq 1/2$. Thus, $\disagree_u(\cP, \varnothing, E^-)\leq 2y_u$ for
every $u$, and $\bE\|\disagree(\cP, \varnothing, E^-)\|_q\leq 2\|y\|_q\leq 2LP$.
This completes the proof.
\end{proof}

\section{Correlation Clustering on Complete Graphs}\label{sec:cor-clust-complete}

In this section, we present our algorithm for Correlation Clustering on complete graphs and its analysis. Our algorithm
achieves an approximation ratio of $5$ and is an improvement over the approximation ratio of $7$ by \citet*{CGS17}.

\subsection{The Algorithm}
Our algorithm is based on rounding an optimal solution to the convex relaxation~(P). Recall that for complete graphs, we can get a simpler relaxation by removing the variables $z$ in our convex programming formulation. We start with considering the entire vertex set of
unclustered vertices. At each step $t$ of the algorithm, we select a subset of vertices as a cluster $C_t$ and remove it
from unclustered vertices. Thus, each vertex is assigned to a cluster exactly once and is never removed from a cluster once it is assigned.

For each vertex $w \in V$, let $\Ball(w,\rho) = \{u \in V : x_{uw} \leq \rho\}$ be the set of vertices within a distance of $\rho$ from $w$.
For $r = 1/5$ the quantity $r - x_{uw}$ where $u \in Ball(w,r)$ represents the distance from $u$ to the boundary of the ball of
radius $1/5$ around $w$. Let $V_t \subseteq V$ be the set of unclustered vertices at step $t$, and define
$$L_t(w) = \sum_{u \in \Ball(w,r) \cap V_t} r - x_{uw}.$$
At each step $t$, we select the vertex $w_t$ that maximizes the quantity $L_t(w)$ over all unclustered vertices $w\in V_t$ and select the set $Ball(w_t,2r)$ as a cluster.
We repeat this step until all the nodes have been clustered. A pseudo-code for our algorithm is given in Figure~\ref{fig:Alg2}.

\begin{figure}

\notarxiv{\begin{center}}
\begin{algorithm}
\openLP
\smallskip

\noindent \textbf{Input: } Optimal solution $x$ to the linear program (P).\notarxiv{\\}

\noindent \textbf{Output: } Clustering $\calC$.
\medskip
\begin{enumerate}
  \item Let  $V_0 = V$, $r = 1/5$, $t = 0$.
  \item \textbf{while} ($V_t \neq \varnothing$)
  \begin{itemize}
\item Find $w_t = \argmax\limits_{w \in V_t} L_t(w)$.
\item Create a cluster $C_t = \Ball(w_t,2r)\cap V_t$.
\item Set $V_{t+1} = V_t \setminus C_t$ and $t = t+1$.
\end{itemize}
\item Return  $\calC = (C_0,\dots, C_{t-1})$.
\end{enumerate}
\end{algorithm}
\notarxiv{\end{center}}
\closeLP

  \caption{Algorithm for Correlation Clustering on complete graphs.}\label{fig:Alg2}\label{alg:corelation-complete}
\end{figure}

\subsection{Analysis}
In this section, we present an analysis of our algorithm.

\begin{theorem}\label{thm:5-apx-main}
Algorithm 2 gives a $5$-approximation for Correlation Clustering on complete graphs.
\end{theorem}

For an edge $(u,v) \in E$, let $LP(u,v)$ be the LP cost of the edge $(u,v)$:
$\lp{u,v} = x_{uv} $ if $(u,v) \in E^+$ and  $\lp{u,v} = 1 - x_{uv}$ if $(u,v) \in E^-$. Let $\alg{u,v} = \ONE( (u,v) \text{ is in disagreement )}$.

Define
$$\pft{u} = \sum_{(u,v) \in E} \lp{u,v} - r \sum_{(u,v) \in E} \alg{u,v},$$
where $r=1/5$. We show that for each vertex $u \in V$, we have $\pft{u} \geq 0$ (see Lemma~\ref{lem:pft} below) and,
therefore, the number of disagreeing edges incident to $u$ is upper bounded by $5y(u)$:
$$ALG(u) = \smashoperator[r]{\sum_{v:(u,v) \in E}} \alg{u,v} \leq \frac{1}{r} \smashoperator[r]{\sum_{v:(u,v) \in E}} \lp{u,v} = 5y(u).$$
Thus, $\|ALG\|_q \leq 5 \|y\|_q$ for any $q\geq 1$. Consequently, the approximation ratio of the algorithm is at most $5$ for any norm $\ell_q$.

\begin{lemma}\label{lem:pft}
For every $u\in V$, we have $\pft{u} \geq 0$.
\end{lemma}
At each step $t$ of the algorithm, we create a new cluster $C_t$ and remove it from the graph. We also remove all
edges with at least one endpoint in $C_t$. Denote this set of edges by
$$\Delta E_t=\{(u,v): u\in C_t \text{ or } v \in C_t\}.$$
Now let
$$
\prft{u,v}{t} =  \begin{cases}\lp{u,v} - r \alg{u,v},& \text{if } (u,v)\in \Delta E\\0,&\text{otherwise}\end{cases}.
$$

\begin{equation}\label{eq:for-profit-u}
\prft{u}{t} = \sum_{v\in V_t}\prft{u,v}{t}
= \smashoperator[r]{\sum_{(u,v) \in \Delta E_t}} \lp{u,v} - r \smashoperator[r]{\sum_{(u,v) \in \Delta E_t}}\alg{u,v}.
\end{equation}
As all sets $\Delta E_t$ are disjoint, $\pft{u} = \sum_t \prft{u}{t}$. Thus, to prove Lemma~\ref{lem:pft},
it is sufficient to show that $\prft{u}{t}\geq 0$ for all $t$. Note that we only need to consider $u\in V_t$ as $\prft{u}{t} = 0$ for
$u\notin V_t$.

Consider a step $t$ of the algorithm and vertex $u\in V_t$. Let $w = w_t$ be the center of the cluster chosen at this step.
First, we show that since the diameter of the cluster $C_t$ is $4r$, for all negative edges $(u,v) \in E^-$ with $u,v \in C_t$, we can charge the cost of
disagreement to the edge itself, that is, $\prft{u,v}{t}$ is nonnegative for $(u,v)\in E^-$ (see Lemma~\ref{cl:neg-edge-profit-nenneg}). We then consider two cases:
$x_{uw}\in [0, r]\cup [3r,1]$ and $x_{uw}\in (r,3r]$.

The former case is fairly simple since disagreeing positive edges $(u,v)\in E^+$ (with  $x_{uw}\in [0, r]\cup [3r,1]$) have a ``large'' LP cost. In Lemma~\ref{lem:0r} and Lemma~\ref{lem:r1},
we prove that the cost of disagreement can be charged to the edge itself and hence $\prft{u}{t} \geq 0$.

We then consider the latter case. For vertices $u$ with $x_{uw} \in (r, 3r]$,
$\prft{u,v}{t}$ for some disagreeing positive edges $(u,v)$ might be negative. Thus, we split the profit at step $t$ for such vertices $u$ into the profit
they get from edges $(u,v)$ with $v$ in $\Ball(w,r)\cap V_t$ and from edges with $v$ in $V_t\setminus \Ball(w,r)$. That is,
$$
\prft{u}{t} = \underbrace{\sum_{v\in \Ball(w,r)}\prft{u,v}{t}}_{P_{high}(u)} + \underbrace{\sum_{v\in V_t\setminus \Ball(w,r)}\prft{u,v}{t}}_{P_{low}(u)}.
$$
Denote the first term by $P_{high}(u)$ and the second term by $P_{low}(u)$. We show that $P_{low}(u)\geq -L_t(u)$ (see Lemma~\ref{lem:PLow-Ltu})
and $P_{high}\geq L_t(w)$  (see Lemma~\ref{lem:PHigh-Ltw}) and
 conclude that $\prft{u}{t} = P_{high}(u) + P_{low}(u)\geq L_t(w)-L_t(u)\geq 0$ since $L_t(w) = \max_{w'\in V_t} L_t(w') \geq L_t(u)$.

In the following claim, we show that we can charge the cost of disagreement of a negative edge to the edge itself.
\begin{claim}\label{cl:neg-edge-profit-nenneg}
For a negative edge $(u,v)\in E^-$, $\prft{u,v}{t}$ is always nonnegative.
\end{claim}
\begin{proof}
The only case when $(u,v)$ is in disagreement is when both $u$ and $v$ belong to the new cluster. In this
case, they lie in the ball of radius $2r$ around $w$ (and thus $x_{uw}, x_{vw} \leq 2r$). Thus the distance $x_{uv}$ between them is at most $4r$ (because $x_{uv} \leq x_{uw} + x_{vw} \leq 4r$). The LP cost of the edge $(u,v)$ is at least $LP(u,v) = 1 - x_{uv} \geq 1- 4r = r$. Thus, $ \prft{u,v}{t} = LP(u,v)-r ALG(u,v) = LP(u,v)- r \geq 0$.
\end{proof}

In Lemma~\ref{lem:0r} and Lemma~\ref{lem:r1}, we consider the case when $x_{uw} \in [0,r] \cup (3r, 1]$.
\begin{lemma}\label{lem:0r}
If $x_{uw}\leq r$, then $\prft{u,v}{t}\geq 0$ for all $v\in V_t$.
\end{lemma}
\begin{proof}
If $x_{uw}\in E^-$, then $\prft{u,v}{t}\geq 0$ by Claim~\ref{cl:neg-edge-profit-nenneg}. Assume that $x_{uw}\in E^+$. Since
$x_{uw}\leq r$, $u$ belongs to the cluster $C_t$. Thus, $(u,v)$ disagrees only if $v$ does not belong to that cluster.
In this case, $x_{wv}\geq 2r$ and by the triangle inequality $x_{uv}\geq x_{vw} - x_{uw}\geq r$. Therefore,
$\prft{u,v}{t} = x_{u,v}-r \geq 0$.
\end{proof}
\begin{lemma}\label{lem:r1}
If $x_{uw}\geq 3r$, then $\prft{u,v}{t}\geq 0$ for all $v \in V_t$.
\end{lemma}
\begin{proof}
As in the previous lemma, we can assume that $x_{uw}\in E^+$. If $x_{uw}\geq 3r$,
then $u$ does not belong to the new cluster $C_t$. Thus, $(u,v)$ disagrees only if $v$ belongs to $C_t$.
In this case, $x_{wv}\leq 2r$ and by the triangle inequality $x_{uv}\geq x_{uw} - x_{vw}\geq r$. Therefore,
$\prft{u,v}{t} = x_{u,v}-r \geq 0$.
\end{proof}

We next consider $u$ such that $x_{uw} \in (r, 3r]$. First, we show that the profit we obtain from every edge $(u,v)$ with $v \in \Ball(w,r)$ is
at least $r - x_{vw}$, regardless of whether the edge is positive or negative.
\begin{claim}\label{claim:prof-from-core-v}
If $x_{uw} \in (r,3r]$ and $v \in \Ball(w,r)\cap V_t$, then $\prft{u,v}{t}\geq r-x_{vw}$.
\end{claim}
\begin{proof}
First consider $u$ such that $x_{uw} \in (r, 2r]$. Note that $x_{uv} \geq x_{uw} - x_{vw} \geq r - x_{vw}$. Moreover, $x_{uv} \leq x_{uw} + x_{vw} \leq 2r + x_{vw}$. Thus, if $(u,v) \in E^+$, then $\prft{u,v}{t} \geq r - x_{vw}$. Otherwise, $\prft{u,v}{t} \geq (1 - 2r - x_{vw}) - r \geq 2r - x_{vw}$.

For $u \in (2r, 3r]$, note that $x_{uv} \geq x_{uw} - x_{vw} \geq 2r - x_{vw}$. Moreover, $x_{uv} \leq x_{uw} + x_{vw} \leq 3r + x_{vw}$.
Thus, if $(u,v) \in E^+$, then $\prft{u,v}{t} \geq (2r - x_{vw}) - r \geq r - x_{vw}$. Otherwise, $\prft{u,v}{t} \geq (1 - 3r - x_{vw}) \geq 2r - x_{vw}$.
\end{proof}

Using the above claim, we can sum up the profits from all vertices $v$ in $\Ball(w, r)$ and lower bound $P_{high}(u)$ as follows.
\begin{lemma}\label{lem:PHigh-Ltw}
If $x_{uw}\in (r,3r]$, then $P_{high}(u) \geq L_t(w)$.
\end{lemma}
\begin{proof}
By Claim~\ref{claim:prof-from-core-v}, we have $\prft{u,v}{t}\geq r-x_{vw}$ for all $v\in V_t$. Thus,
$$
P_{high}(u) = \sum_{v\in \Ball(w,r)\cap V_t}\prft{u,v}{t}\geq \sum_{v\in \Ball(w,r)\cap V_t}r-x_{vw} = L_t(w).
$$
\end{proof}

We now lower bound $P_{low}(u)$. To this end. we estimate each term $\prft{u,v}{t}$ in the definition of $P_{low}$.
\begin{claim}\label{claim:prof-uv-lower-bound}
If $x_{uw} \in (r,3r]$ and $v \in V_t \setminus \Ball(w,r)$, then $\prft{u,v}{t}\geq \min(x_{uv} - r, 0)$.
\end{claim}
\begin{proof}
By Claim~\ref{cl:neg-edge-profit-nenneg}, if $(u,v)$ is a negative edge, then $\prft{u,v}{t} \geq 0$. The profit is $0$ if $x_{uv}\notin \Delta E_t$ (i.e., neither $u$ nor $v$ belong
to the new cluster). So let us assume that $(u,v)$ is a positive edge in $\Delta E_t$. Then, the profit obtained from $(u,v)$ is $x_{uv}$ if $(u,v)$ is in agreement
and $x_{uv} - r$ if $(u,v)$ is in disagreement. In any case, $\prft{u,v}{t} \geq x_{uv} - r \geq \min(x_{uv} - r, 0)$.
\end{proof}

Lemma~\ref{lem:PLow-Ltu} is an immediate corollary of Claim~\ref{claim:prof-uv-lower-bound}.
\begin{lemma}\label{lem:PLow-Ltu}
If $x_{uw}\in (r,3r]$, then $P_{low}(u) \geq -L_t(u)$.
\end{lemma}
\begin{proof}
By Claim~\ref{claim:prof-uv-lower-bound}, we have $\prft{u,v}{t}\geq \min(x_{uv} - r,0)$ for all $v\in V_t$. Thus,
\begin{align*}
P_{low}(u) &= \sum_{v\in V_t\setminus \Ball(w,r)}\prft{u,v}{t}\\
&\geq  \sum_{v\in V_t\setminus \Ball(w,r)} \min(x_{uv} - r,0)\\
&\overset{a}{\geq} \;\;\;\;\;\sum_{v\in V_t} \min(x_{uv} - r,0) \\
&\overset{b}{=} \sum_{v\in \Ball(u,r) \cap V_t} x_{uv} - r \\
&= - L(u).
\end{align*}
Here we used that (a) all terms $\min(x_{uv} - r,0)$ are nonpositive, and (b) $\min(x_{uv} - r, 0) = 0$ if $v\notin \Ball(u,r)$.
\end{proof}

This finishes the proof of Lemma~\ref{lem:pft}.
\pagebreak

\section{Correlation Clustering with AKS Min Max Objective}

In this section, we present our improved algorithm for Correlation Clustering with AKS Min Max Objective. Our algorithm produces a clustering of cost at most $(2 + \varepsilon) OPT$, which improves upon the bound of $O(\log n)$ given by~\citet*{AKS}.

 For a subset $S \subseteq V$ of vertices, we use $\cost^+(S)$ to refer to the weight of positive edges ``associated'' with $S$ that
are in disagreement. These are the edges with exactly one end point in $S$. Thus, $\cost^+(S) = \sum_{(u,v) \in E^+, u \in S, v \not\in S} w_{uv}$. Similarly, we use $\cost^-(S)$ to refer to the weight of dissimilar edges ``associated'' with $S$ that are in disagreement. These are the edges with both endpoints in $S$. Thus, $\cost^-(S) = \sum_{(u,v) \in E^-, u,v \in S} w_{uv}$. The total cost of the set $S$ is $\cost(S) = \cost^+(S) + \cost^-(S)$.

Similar to the algorithm of~\citet{AKS}, our algorithm works in two phases. In the first phase,
the algorithm covers  all vertices of the graph with (possibly overlapping) sets $S_1,\dots, S_k$ such that the cost of each set $S_i$ is at most $2OPT$ (i.e., $\cost(S_i)\leq 2OPT$ for each $i \in \{1,\dots, k\}$).
In the second phase, the algorithm finds sets $P_1,\dots, P_k$ such that:
(1) $P_1,\dots, P_k$ are disjoint and cover the vertex set;
(2) $P_i \subseteq S_i$ (and, consequently, $\cost^-(P_i)\leq \cost^-(S_i)$);
(3) $\cost^+(P_i)\leq (1+\varepsilon)\cost^+(S_i)$.

The sets $P_1,\dots, P_k$ are obtained from $S_1, \ldots,S_k$ using an uncrossing procedure of~\cite{BFKMNS}. Hence the clustering that is output is $\calP=(P_1,\dots, P_k)$. The improvement in the approximation factor comes from the first phase of the algorithm.

\subsection{The algorithm}\label{alg:mincostcover}

At the core of our algorithm is a simple subproblem:
For a given vertex $z \in V$, find a subset $S \subseteq V$ containing $z$ such that $\cost(S)$ is minimized.
We solve this subproblem using a linear programming relaxation, which is formulated as follows:
The LP has a variable $x_u$ for each vertex $u\in V$. In the intended integral
solution, we have $x_u = 1$ if $u$ is in the set $S$, and $x_u = 0$, otherwise. That is, $x_u$ is the indicator of the event
 ``$u\in S$''. The LP has only one constraint: $x_z = 1$. A complete description of the LP can be found in Figure~\ref{fig:cover-z}. In Claim~\ref{clm:valid-lp-subproblem} we show that this LP is indeed a valid relaxation for our subproblem.

\begin{figure}
  \centering
\begin{equation*}
\begin{array}{ll@{}llr}
\text{minimize}  & \displaystyle \sum_{(u,v)\in E^+} w_{uv}|x_u-x_v| + \sum_{(u,v)\in E^-} w_{uv}(x_u + x_v -1)^+\\
\text{subject to}
&x_z = 1\\
&0\leq x_u\leq 1\text{ for all } u \in V
\end{array}
\end{equation*}
Here, we use notation $(t)^+ = \max(t,0)$.
  \caption{LP relaxation for covering $z$ with a low cost set $S$.}\label{fig:cover-z}
\end{figure}

\begin{claim}\label{clm:valid-lp-subproblem}
The LP relaxation described in Figure~\ref{fig:cover-z} is a valid relaxation for the subproblem.
\end{claim}
\begin{proof}
Let us verify that this is a valid relaxation for the problem. As we discussed above,
in the intended integral solution, we have $x_u = 1$ if $u$ is in the set $S$, and $x_u = 0$, otherwise.
That is, $x_u$ is the indicator of the event ``$u\in S$''.

Consider a positive edge $(u,v)\in E^+$. In the integral solution, $|x_u-x_v|=1$ if and only if one of the vertices $u$ or $v$ is in $S$ and the other one is not.
In this case, the edge $(u,v)$ is in disagreement with $S$. Now, consider
a negative edge $(u,v)\in E^-$. In the integral solution, $(x_u + x_v -1)^+ = 1$ if and only if both $u$ and $v$ are
in $S$. Again, in this case,  the edge $(u,v)$ is in disagreement with $S$. Thus, this LP is a relaxation for our
problem.

Note that we can linearize the $|\cdot|$ and $(\cdot)^+$ terms in the objective as follows. We can replace terms of the type $|x_u - x_v|$ with variables $\mu_{uv}$ and introduce the constraints $\mu_{uv} \geq (x_u - x_v)$ and $\mu_{uv}  \geq (x_v - x_u)$. Similarly, we can replace terms of the type $(x_u + x_v -1)^+$ with variables $\eta_{uv}$ and introduce the constraints 
$\eta_{uv} \geq (x_u + x_v -1)$ and $\eta_{uv} \geq 0$. It is easy to see that the minimum values for the variables $\mu_{uv}$ and $\eta_{uv}$ is attained at $|x_u - x_v|$ and $(x_u + x_v -1)^+$ respectively.
\end{proof}

 We are now ready to present our algorithm.

\noindent \textbf{Algorithm (Find minimum cost set):} For each $t\in [0,1]$, define a threshold
set $S_t$ as $S_t=\{u:x_u \geq t\}$. There are at most $n$ such distinct sets $S_t$ (since the set $\{x_u:u\in V\}$
contains at most $n$ elements). Our algorithm picks a set $S_t$ that minimizes
$\cost(S_t)$ and outputs it.

\begin{lemma}\label{lm:2approx-round}
The algorithm described above finds a set of cost at most $2LP$, where $LP$ is the
cost of the $LP$ solution.
\end{lemma}
\begin{proof}
We show that if we pick $t$ uniformly at random in $[1/2,1]$
then the expected cost of a random set $S_t$ is at most $2LP$. Consequently,
the minimum cost of set $S_t$ for $t\in[0,1]$ is at most $2LP$ and, hence,
the algorithm returns a solution of cost at most $2 LP$.

The probability that a positive edge $(u,v)\in E^+$ is in disagreement
with $S$ equals the probability that random $t$ lies between $x_u$ and $x_v$,
which is at most $2|x_u-x_v|$ (since $|x_u-x_v|$ is the length of the
interval $[x_u,x_v]$ and $2$ is the density of the random variable $t$
on the interval $[1/2,1]$). That is, this probability is upper bounded by
twice the LP cost of the edge $(u,v)$.
The probability that a negative edge $(u,v)\in E^-$ is
in disagreement with $S$ equals the probability that $t\leq \min(x_u,x_v)$
which is $2(\min(x_u,x_v) - 1/2)^+ = (2\min(x_u,x_v) -1)^+$.
This probability is upper bounded by the LP cost
of the negative edge $(u,v)$ (i.e., $(x_u+x_v-1)^-$). This concludes the
proof.
\end{proof}

 Thus, to obtain a cover of all the vertices, we pick yet uncovered vertices $z\in V$ one by one and for each $z$, find a set $S(z)$ as described above.
Then, we remove those sets $S(z)$ that are completely covered by other sets. The obtained family of sets $\calS=\{S(z)\}$ satisfies the following
properties: (1) Sets in $\calS$ cover the entire set $V$; (2) $\cost(S) \leq 2OPT$ for each $S\in \calS$; (3) Each set $S\in \calS$ is not covered by the other
sets in $\calS$ (that is, for each $S \in \calS$, $S \not\subset \cup_{S' \in (\calS \setminus \{S\})} S'$). However, sets $S$ in $\calS$ are not necessarily disjoint.

Following~\cite{AKS}, we then apply an uncrossing procedure developed by~\cite{BFKMNS} to the sets $S_i$ in $\calS$ and obtain
disjoint sets $P_i$. Each set $P_i$ is a subset of $S_i$ and therefore $\cost^-(P_i)\leq \cost^-(S_i)$. Moreover, $\cost^+(P_i)\leq \cost^+(S_i) + \varepsilon OPT$ (see Section ~\ref{sec:aks-uncrossing}).
Hence, $P_1,\dots, P_k$ is a $2(1+\varepsilon)$ approximation for  Correlation Clustering with the AKS Min Max objective.

\bibliography{corr-clust}
\bibliographystyle{plainnat}

\appendix

\section{Uncrossing Overlapping Sets}\label{sec:aks-uncrossing}
For completeness, we present here a proof of the following lemma from~\cite{BFKMNS}. Denote by $\delta(S)$ the set of all positive edges leaving set
$S$ in graph $G$. Then, $\cost^+(S)=w(\delta(S))$.

\begin{lemma}[Uncrossing argument in~\cite{BFKMNS}]
There exists a polynomial-time algorithm that given a weighted graph $G=(V,E)$, a family of sets $S_1,\dots S_k$ that covers all
vertices in $G$, and a parameter $\varepsilon = 1/poly(n)$, finds disjoint sets $P_1,\dots,P_k$ covering $V$ such that for each $i$:
\begin{enumerate}
  \item $P_i\subset S_i$; and
  \item $w(\delta(P_i))\leq w(\delta(S_i)) + \varepsilon \max_j w(\delta(S_j))$.
\end{enumerate}
\end{lemma}
\begin{proof}
Let us first describe the uncrossing algorithm from the paper~~\cite{BFKMNS}. Initially, the algorithm
sets $P^0_i = S_i \setminus \cup_{j< i}S_j$ for each $i\in\{1,\dots, k\}$. Then, at every step $t$, it finds a set $P_i^t$ violating the
desired bound
\begin{equation}\label{eq:bound-uncrossing}
w(\delta(P^t_i))\leq w(\delta(S_i)) + \varepsilon \max_j w(\delta(S_j))
\end{equation}
and updates all sets as follows: $P^{t+1}_i = S_i$; and $P^{t+1}_j = P^t_j \setminus S_i$.
The algorithm terminates and outputs sets $P^t_i$ when bound~(\ref{eq:bound-uncrossing}) holds for all sets $P^t_i$.

It easy to see that the following loop invariants hold at every step of the algorithm: (1) each $P^t_i$ is a subset of $S_i$; (2) sets $P^t_i$ are disjoint; (4) sets $P^t_i$ cover all vertices in $V$. It is also immediate that when or if the algorithm terminates
sets $P^t_i$ satisfy~(\ref{eq:bound-uncrossing}). We only need to check that the algorithm stops in polynomial time.

Let $B = \max_j w(\delta(S_j))$. Define a potential function $\varphi(t) = \sum_{i=1}^k w(\delta(P_i))$. Observe that initially
$\varphi(0)\leq 2\sum_i w(\delta(S_i))$, since every edge cut by the partition $(P_1,\dots, P_k)$ belongs to some $S_i$.
Since, $w(\delta(S_i))\leq B$ for all $i$, we have $\varphi(0)\leq 2kB$. We will show that at every step of the algorithm
$\varphi(t)$ decreases by at least $2\varepsilon B$ and thus the algorithm terminates in at most $k/\varepsilon$ steps.

Consider step $t$ of the algorithm. Suppose that at this step of the algorithm, set $P^t_i$ violated the constraint and thus it was replaced by
$S_i$. Write,
\begin{align*}
\varphi(t+1) - \varphi(t) &= \Big(w(\delta(S_i)) -  w(\delta(P^t_i))\Big) + \sum_{j\neq i}(w(\delta(P^{t+1}_i)) - w(\delta(P^{t}_i)))\\
&=\Big(w(\delta(S_i)) -  w(\delta(P^t_i))\Big) + \sum_{j\neq i}(w(\delta(P^{t}_i\setminus S_i)) - w(\delta(P^{t}_i))).
\end{align*}
Observe that for every two subsets of vertices $P$ and $S$ we have the following inequality:
\begin{align*}
w(\delta(P\setminus S)) - w(\delta(P)) &= \phantom{-} \Big(w(E(P \setminus S, V\setminus P)) + w(E(P \setminus S, P\cap  S))\Big) \\
&\phantom{=} - \Big(w(E(P \setminus S, V\setminus P)) + w(E(P\cap  S, V \setminus P))\Big)\\
&= \phantom{-} w(E(P\cap  S, P \setminus S)) - w(E(P\cap  S, V \setminus P))\\
&\leq \phantom{-} w(E(P\cap  S, P \setminus S)) - w(E(P\cap  S, S \setminus P))\\
&=\phantom{-}\Big(w(E(P\cap S, P \setminus S)) + w(E(S \setminus P, P \setminus S))\Big)  \\ &\phantom{=} -
\Big(w(E(P\cap S, S \setminus P)) + w(E(P \setminus S, S \setminus P))\Big)\\
&=\phantom{-}w(E(S, P \setminus S)) -w(E(P, S \setminus P)).
\end{align*}
Also, note that $P^t_i \subset S_i \setminus P_j^t$ (since $P^t_i \subset S_i$ and all $P^t_j$ are disjoint). Consequently,
$w(E(P^t_i, P^t_j)) \leq  w(E(S_i \setminus P_j^t, P^t_j))$.
Therefore,
\begin{align*}
\varphi(t+1) - \varphi(t)
&= \Big(w(\delta(S_i)) -  w(\delta(P^t_i))\Big) + \sum_{j\neq i} w(E(S_i, P_j^t \setminus S_i)) -w(E(P_j^t, S_i \setminus P_j^t))\\
&\leq \Big(w(\delta(S_i)) -  w(\delta(P^t_i))\Big) + \sum_{j\neq i} w(E(S_i, P_j^t \setminus S_i)) -w(E(P_j^t, P_i^t)).
\end{align*}
Using again that the sets $P_j^t$ partition $V$ into disjoint pieces, we get
\begin{align*}
\varphi(t+1) - \varphi(t)
&\leq \Big(w(\delta(S_i)) -  w(\delta(P^t_i))\Big) + \sum_{j\neq i} w(E(S_i, P_j^t \setminus S_i)) -w(E(P_j^t, P_i^t))\\
&= \Big(w(\delta(S_i)) -  w(\delta(P^t_i))\Big) + w(E(S_i, \cup_{j\neq i} P_j^t \setminus S_i)) -w(E( \cup_{j\neq i} P_j^t, P_i^t))\\
&=\Big(w(\delta(S_i)) -  w(\delta(P^t_i))\Big) + \underbrace{w(E(S_i, V \setminus S_i))}_{=\delta(S_i)} - \underbrace{w(E(V\setminus P_i^t, P_i^t))}_{=\delta(P_i^t)}\\
&=2 \Big(w(\delta(S_i)) -  w(\delta(P^t_i))\Big) \leq -2\varepsilon B.
\end{align*}
This concludes the proof.
\end{proof}

\section{Integrality gap}

In this section, we present an integrality gap example for the convex program (P). We describe an instance of the $\ell_q$ $s-t$ cut problem on $\Theta(n)$ vertices that has an integrality gap of $\Omega(n^{\frac{1}{2} - \frac{1}{2q}})$. In our integrality gap example, we describe a layered graph with $\Theta(n^\frac{1}{2})$ layers, with each layer consisting of a complete bipartite graph on $\Theta(n^\frac{1}{2})$ vertices. Between each layer $i$ and $i+1$, there is a terminal $s_i$ which connects these two layers. Finally, the terminals $s$ and $t$ are located at opposite ends of this layered graph. We will observe that for any integral cut separating $s$ and $t$, there will be at least one vertex such that a large fraction of the edges incident to it are cut. We will show that there is a corresponding fractional solution that is cheaper compared to any integral cut as the fractional solution can ``spread'' the cut equally across the layers, thus not penalizing any individual layer too harshly. In doing so, we will prove the following theorem,

\begin{theorem}
The integrality gap for the convex relaxation (P) is $\Omega(n^{\frac{1}{2}-\frac{1}{2q}})$.
\end{theorem}

\begin{proof}
We now give a more formal description of the layered graph discussed above. The construction has two parameters $a$ and $b$, so we will call such a graph $G_{a,b}$. The graph consists of $b$ layers with each layer consisting of the complete bipartite graph $K_{a,a}$. We refer to layer $i$ of the graph as $G^i_{a,b}$ and refer to the left and right hand of the bipartition as $L(G^i_{a,b})$ and $R(G^i_{a,b})$ respectively. In addition to these layers, the graph consists of $b+1$ terminals $\{s, t, s_1, \ldots, s_{b-1}\}$ (we will refer to $s$ as $s_0$ and $t$ as $s_b$ interchangeably). For each $i \in \{1, \ldots, b-1\}$, the vertex $s_i$ is connected to all the vertices in $R(G^i_{a,b})$ and $L(G^{i+1}_{a,b})$. Finally, $s$ is connected to all the vertices in $L(G^1_{a,b})$ and $t$ is connected to all the vertices in $R(G^b_{a,b})$.

Consider any integral cut separating $s$ and $t$ in the graph $G_{a,b}$. Any such cut must disconnect at least one pair of consecutive terminals (if all pairs of consecutive terminals are connected, then $s$ is still connected to $t$). Thus let $j \in \{0, 1, \ldots, b\}$ be such that $s_{j-1}$ is disconnected from $s_{j}$ and consider the subgraph induced on $\{s_{j-1} \cup s_{j} \cup G^j_{a,b}\}$. We will show that this induced subgraph contains a vertex such that $\Omega(a^\frac{1}{2})$ of its incident edges are cut. Intuitively, since $s_{j-1}$ is separated from $s_j$, if the majority of the edges incident to $s_{j-1}$ and $s_j$ are not cut, then $s_{j-1}$ and $s_{j}$ have many neighbors in $L(G^{j}_{a,b})$ and $R(G^{j}_{a,b})$ respectively. As $G^{j}_{a,b}$ is highly connected, in order for $s_{j-1}$ to be separated from $s_j$, there must be a vertex in $G^j_{a,b}$ with many incident edges which are cut. If $cut(s_{j-1})$ or  $cut(s_{j})$ is at least $a/2$, then we are done. Otherwise, $s_j$ is connected to at least $a/2$ vertices in $R(G^{j}_{a,b})$, so every $u$ adjacent to $s_{j-1}$ must have at least $a/2$ incident edges which are cut. Therefore, $OPT^q \geq \Omega(a^q)$.

We now present a fractional cut separating $s$ and $t$. If an edge $e$ connects $s_i$ to a vertex in $R(G^i_{a,b})$ for some $i \in \{1, \ldots, b\}$, set the length of the edge to be $1/b$; otherwise set the edge length to be $0$. We let $x_{uv}$ be the shortest path metric in this graph. It is easy to see that such a solution is feasible. We now analyze the quality of this solution. 
For each $i \in \{1, \ldots, b\}$, we have $y_{s_i} = a/b$ and for each $u \in R(G^i_{a,b})$, we have $y_u = 1/b$. Thus 
$$LP^q = ab\Big(\frac{1}{b}\Big)^q + b\Big(\frac{a}{b}\Big)^q.$$ 
If $b>a$, then 
$$LP^q \leq ab\Big(\frac{1}{b}\Big) + b\Big(\frac{a}{b}\Big) = 2a$$ 
and if $b > a$, then 
$$LP^q \leq ab\Big(\frac{1}{b}\Big) + b\Big(\frac{a}{b}\Big)^q \leq a^q\Big(a^{-(q-1)} + b^{-(q-1)}\Big).$$

Setting $a = b = \Omega({n^\frac{1}{2}})$ gives 
$$\frac{OPT^q}{LP^q} = \Omega\Big(n^{\frac{q}{2}-\frac{1}{2}}\Big),$$
so the integrality gap is $\frac{OPT}{LP} = \Omega(n^{\frac{1}{2}-\frac{1}{2q}})$.
\end{proof}

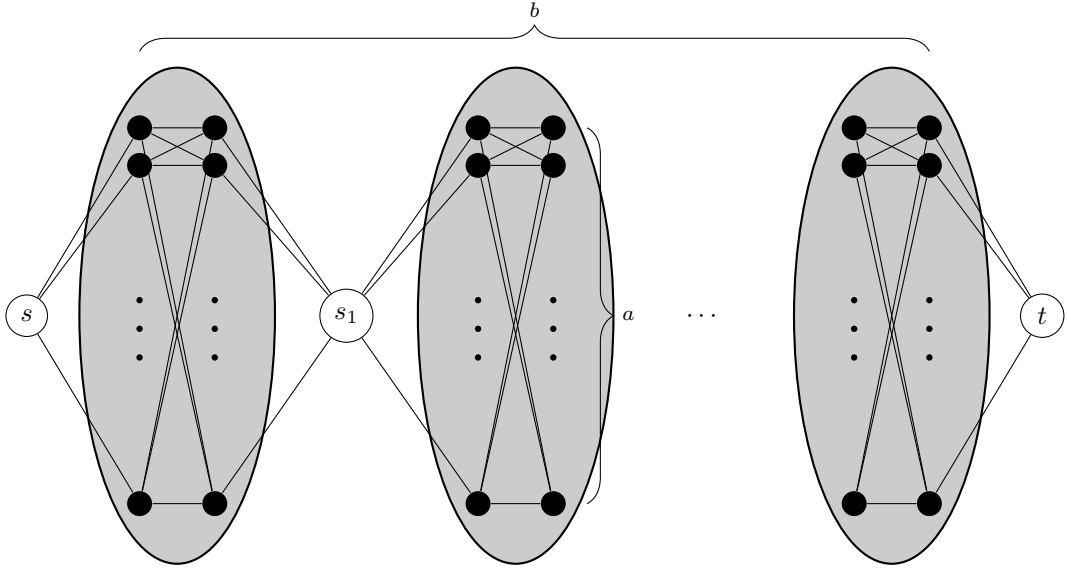
\begin{figure*}
  \centering
  \begin{center}
\begin{tikzpicture}

\fill[fill=black!20!white,draw=black,thick] (-12,0) ellipse (1.3 and 3.3);
\node[shape=circle,fill=black] (j1) at (-12.5,2.5) {};
\node[shape=circle,fill=black] (j2) at (-12.5,2) {};
\node[shape=circle,fill=black] (j3) at (-12.5,-2.5) {};

\node[shape=circle,fill=black] (j4) at (-11.5,2.5) {};
\node[shape=circle,fill=black] (j5) at (-11.5,2) {};
\node[shape=circle,fill=black] (j6) at (-11.5,-2.5) {};

\draw[draw=black]  (j1) edge node {} (j4);
\draw[draw=black]  (j1) edge node {} (j5);
\draw[draw=black]  (j1) edge node {} (j6);

\draw[draw=black]  (j2) edge node {} (j4);
\draw[draw=black]  (j2) edge node  {} (j5);
\draw[draw=black]  (j2) edge node  {} (j6);

\draw[draw=black]  (j3) edge node  {} (j4);
\draw[draw=black]  (j3) edge node  {} (j5);
\draw[draw=black]  (j3) edge node  {} (j6);

\path (j2) -- (j3) node [black, font=\Huge, midway, sloped] {$\dots$};
\path (j5) -- (j6) node [black, font=\Huge, midway, sloped] {$\dots$};

\node[shape=circle, fill=white, draw=black] (s) at (-14,0) {$s$};
\node[shape=circle, fill=white, draw=black] (s1) at (-9.75,0) {$s_1$};
\node[shape=circle, fill=white, draw=black] (t) at (-0.5,0) {$t$};

\fill[fill=black!20!white,draw=black,thick] (-7.5,0) ellipse (1.3 and 3.3);
\node[shape=circle,fill=black] (k1) at (-8,2.5) {};
\node[shape=circle,fill=black] (k2) at (-8,2) {};
\node[shape=circle,fill=black] (k3) at (-8,-2.5) {};

\node[shape=circle,fill=black] (k4) at (-7,2.5) {};
\node[shape=circle,fill=black] (k5) at (-7,2) {};
\node[shape=circle,fill=black] (k6) at (-7,-2.5) {};

\draw[draw=black]  (k1) edge node {} (k4);
\draw[draw=black]  (k1) edge node {} (k5);
\draw[draw=black]  (k1) edge node {} (k6);

\draw[draw=black]  (k2) edge node {} (k4);
\draw[draw=black]  (k2) edge node  {} (k5);
\draw[draw=black]  (k2) edge node  {} (k6);

\draw[draw=black]  (k3) edge node  {} (k4);
\draw[draw=black]  (k3) edge node  {} (k5);
\draw[draw=black]  (k3) edge node  {} (k6);

\path (k2) -- (k3) node [black, font=\Huge, midway, sloped] {$\dots$};
\path (k5) -- (k6) node [black, font=\Huge, midway, sloped] {$\dots$};

\node at (-5,0) {$\dotsb$};

\fill[fill=black!20!white,draw=black,thick] (-2.5,0) ellipse (1.3 and 3.3);
\node[shape=circle,fill=black] (l1) at (-3,2.5) {};
\node[shape=circle,fill=black] (l2) at (-3,2) {};
\node[shape=circle,fill=black] (l3) at (-3,-2.5) {};

\node[shape=circle,fill=black] (l4) at (-2,2.5) {};
\node[shape=circle,fill=black] (l5) at (-2,2) {};
\node[shape=circle,fill=black] (l6) at (-2,-2.5) {};

\draw[draw=black]  (l1) edge node {} (l4);
\draw[draw=black]  (l1) edge node {} (l5);
\draw[draw=black]  (l1) edge node {} (l6);

\draw[draw=black]  (l2) edge node {} (l4);
\draw[draw=black]  (l2) edge node  {} (l5);
\draw[draw=black]  (l2) edge node  {} (l6);

\draw[draw=black]  (l3) edge node  {} (l4);
\draw[draw=black]  (l3) edge node  {} (l5);
\draw[draw=black]  (l3) edge node  {} (l6);

\path (l2) -- (l3) node [black, font=\Huge, midway, sloped] {$\dots$};
\path (l5) -- (l6) node [black, font=\Huge, midway, sloped] {$\dots$};

\draw[draw=black]  (s) edge node {} (j1);
\draw[draw=black]  (s) edge node {} (j2);
\draw[draw=black]  (s) edge node {} (j3);

\draw[draw=black]  (s1) edge node {} (j4);
\draw[draw=black]  (s1) edge node {} (j5);
\draw[draw=black]  (s1) edge node {} (j6);

\draw[draw=black]  (s1) edge node {} (k1);
\draw[draw=black]  (s1) edge node {} (k2);
\draw[draw=black]  (s1) edge node {} (k3);

\draw[draw=black]  (t) edge node {} (l4);
\draw[draw=black]  (t) edge node {} (l5);
\draw[draw=black]  (t) edge node {} (l6);

\draw [decorate,decoration={brace,amplitude=10pt},xshift=0.45cm,yshift=0pt]
(-7,2.5) -- (-7,-2.5) node [black,midway,xshift=0.55cm]
{\footnotesize $a$};

\draw [decorate,decoration={brace,amplitude=10pt},xshift=0cm,yshift=0.45pt]
(-12.5,3.5) -- (-2,3.5) node [black,midway,yshift=0.55cm]
{\footnotesize $b$};

\end{tikzpicture}
\end{center} 
  \caption{Integrality gap example.}\label{fig:int-gap}
\end{figure*}

\section{Hardness of approximation}
In this section, we prove the following hardness result.

\begin{theorem}
It is NP-hard to approximate the min $\ell_\infty$ s-t cut problem within a factor of $2 - \varepsilon$ for every positive $\varepsilon$.
\end{theorem}
\begin{proof}
The proof follows a reduction from $3$SAT. We will describe a procedure that reduces every instance of a $3$CNF formula $\phi$ to a graph $G_\phi$ such that the minimum $\ell_\infty$ \textit{s-t} cut for $G_\phi$ has a certain value if and only if the formula $\phi$ is satisfiable.

\medskip

\noindent \textbf{Reduction from 3SAT:} Given a $3$CNF instance $\phi$ with $n$ literals and $m$ clauses, we describe a graph $G_\phi$ with $(2 + 4n + 5m)$ vertices
and $(6n + 8m)$ edges. We refer to the vertex and edge set of $G_\phi$ as $V(G_\phi)$ and $E(G_\phi)$. For every literal $x_i,  i \in \{1,\dots, n\}$, we have four nodes,
$x^T_i$, $x^F_i$, $x^{\dagger}_i$ and $\br{x}_i^\dagger$. Additionally, we have a ``False'' and a ``True'' node. For every $i \in \{1,\dots,n\}$, we connect 
``False'' with $x^F_i$ and ``True'' with $x^T_i$ using an infinite weight edge. Both $x^F_i$ and $x^T_i$ are connected to $x^{\dagger}_i$ and $\br{x}_i^\dagger$ using edges
of weight $1$.

For every clause $C$ in $\phi$, we will create a gadget in $G_\phi$ consisting of five nodes. We will refer to the subgraph induced by these nodes as $G_\phi[C]$. Let the 
clause $C = (y_1 \lor y_2 \lor y_3)$.  We have a node in the gadget for each $y_i, i \in \{1,2,3\}$, and two additional nodes $C_a$ and $C_b$. We connect $y_2$ and $y_3$ to $C_b$, 
and $y_1$ and $C_b$ to $C_a$, all using unit weight edges.

We connect the gadget $G_\phi[C]$ for clause $C = (y_1 \lor y_2 \lor y_3)$ to the main graph as follows. For each $i \in \{1,2,3\}$, connect the vertex for the 
literal $y_i$ to the vertex $y^\dagger_i$ with a unit weight edge. Finally connect the node $C_a$ to the ``True'' vertex using an infinite weight edge. An example of
a 3CNF formula $\phi$ and the corresponding $G_\phi$ is given in Figure~\ref{figure:hardness}.

\textbf{Fact 1.} Consider the gadget $G_\phi[C]$ for the clause $C = (y_1 \lor y_2 \lor y_3)$. If all three nodes $y_1, y_2$, and $y_3$ need to be disconnected
from $C_a$, then either $|\cut_{C_a}| = 2$ or $\cut_{C_b} = 2$.  If at most two of the three nodes $y_1, y_2$ and $y_3$ need to be disconnected from $C_a$, then there 
is a cut that separates those nodes from $C_a$ such that both $\cut_{C_a}$ and $\cut_{C_b}$ are at most $1$.

\begin{lemma}
Given a 3CNF formula $\phi$, consider the graph $G_\phi$ constructed according to the reduction described above. The formula $\phi$ is satisfiable
if and only if the minimum $\ell_\infty$ True-False cut $\cP$ for the graph $G_\phi$ has value $1$, that is, $||\cut_\cP||_\infty = 1$.
\end{lemma}

\begin{proof}
\textbf{3SAT $\Rightarrow$ minimum $\ell_\infty$ \textit{True-False cut}}: If the 3CNF formula $\phi$ is satisfiable, then the graph $G_\phi$ has a minimum $\ell_\infty$ \textit{s-t} cut of value exactly $1$. This can be seen as follows. Given a satisfying assignment $x^*$, we will construct a cut $E_\cP$ (and corresponding partition $\cP$) such that for every vertex $u \in V(G_\phi)$, $\cut_\cP(u) \leq 1$. For every $i \in \{1,\dots, n\}$, if $x^*_i$ is True, then include $(x^{\dagger}_i, x_i^F)$ and $(\br{x}_i^\dagger, x_i^T)$ as part of the cut $E_\cP$, else include $(x^{\dagger}_i, x_i^T)$ and $(\br{x}_i^\dagger, x_i^F)$ as part of the cut $E_\cP$. Note that this cuts exactly one edge incident to each vertex $x^{\dagger}_i, x_i^F, \br{x}_i^\dagger$ and $x_i^T$ for $i \in \{1,\dots, n\}$. Since $\phi$ has a satisfiable assignment, each clause $C$ has at least one literal which is True, and hence the node corresponding to this literal is not connected to the vertex False in $G_\phi - E_\cP$. Thus, each clause $C$ has at most two literals that are False, and thus there are at most two False-True paths that go through this gadget. From Fact 1, we can know that we can include edges from $E(G_\phi[C])$ in $E_\cP$ such that both $\cut_\cP(C_a)$ and $\cut_\cP(C_b)$ are at most $1$ and the False-True paths through this gadget are disconnected. Thus, cut $E_\cP$ disconnects True from False such that $||\cut_\cP(G_\phi)||_\infty = 1$.

\medskip

\textbf{minimum $\ell_\infty$ \textit{True-False cut} $\Rightarrow$ 3SAT}:  Let $G_\phi$ be the graph constructed for the 3CNF formula $\phi$ such that there is a cut $E_\cP \subseteq E(G_\phi)$ (and corresponding partition $\cP$) such that $\cP$ separates True from False and $||\cut_\cP(G_\phi)||_\infty = 1$. We will construct a satisfying assignment $x^*$ from the formula $\phi$. Since $\cut_\cP(u) \leq 1$ for every $u \in V(G_\phi)$, none of the $(True, x^T_i)$, $(x^F_i, False)$ edges are part of the cut $\cP$ for $i \in \{1,\dots, n\}$. In order for True to be separated from False, either the edges $(x^{\dagger}_i, x_i^F)$ and $(\br{x}_i^\dagger, x_i^T)$ are part of the cut $E_\cP$, or the edges $(x^{\dagger}_i, x_i^T)$ and $(\br{x}_i^\dagger, x_i^F)$ are part of the cut $E_\cP$. This gives us our assignment; for each $i \in \{1,\dots, n\}$, if $(x^T_i, x^{\dagger}_i) \in E \setminus E_\cP$, then assign $x^*_i$ as True and $\br{x}^*$ as False. Otherwise $(x^F_i, x^{\dagger}_i) \in E \setminus E_\cP$, so assign $x^*_i$ as False and $\br{x}^*$ as True. Now, we show that $x^*$ is a satisfiable assignment for $\phi$. To see this, note that for each clause $C$, there exists at least one literal $y_i$ such that the node corresponding to $y_i$ is still connected to $C_a$. As the cut $E_\cP$ separates True and False, $(y^\dagger_i, y^T_i) \in E \setminus E(G_\phi)$ and hence $y^*_i = $ True. Thus, the assignment $x^*$ is satisfiable for $\phi$.

\end{proof}

Thus, we can conclude Theorem 5.1 from the reduction procedure provided and Lemma 5.2.
\end{proof}

\begin{figure*}
  \centering
  \begin{tikzpicture}

\node[draw,circle,minimum size=1cm,inner sep=0pt] at (-7,0) (T) {TRUE};
\node[draw,circle,minimum size=1cm,inner sep=0pt] at (7,0) (F) {FALSE};

\node[draw,circle,minimum size=0.7cm,inner sep=0pt] at (0,8) (1) {$x^\dagger_1$};
\node[draw,circle,minimum size=0.7cm,inner sep=0pt] at (0,7.2) (1B) {$\br{x}_1^\dagger$};
\node[draw,circle,minimum size=0.4cm,inner sep=0pt] at (-3.5,7.6) (1T) {$x^T_1$};
\node[draw,circle,minimum size=0.4cm,inner sep=0pt] at (3.5,7.6) (1F) {$x^F_1$};

\node[draw,circle,minimum size=0.7cm,inner sep=0pt] at (0,6) (2) {$x^\dagger_2$};
\node[draw,circle,minimum size=0.7cm,inner sep=0pt] at (0,5.2) (2B) {$\br{x}_2^\dagger$};
\node[draw,circle,minimum size=0.4cm,inner sep=0pt] at (-3.5,5.6) (2T) {$x^T_2$};
\node[draw,circle,minimum size=0.4cm,inner sep=0pt] at (3.5,5.6) (2F) {$x^F_2$};

\node[draw,circle,minimum size=0.7cm,inner sep=0pt] at (0,4)  (3) {$x^\dagger_3$};
\node[draw,circle,minimum size=0.7cm,inner sep=0pt] at (0,3.2) (3B) {$\br{x}_3^\dagger$};
\node[draw,circle,minimum size=0.4cm,inner sep=0pt] at (-3.5,3.6) (3T) {$x^T_3$};
\node[draw,circle,minimum size=0.4cm,inner sep=0pt] at (3.5,3.6) (3F) {$x^F_3$};

\node[draw,circle,minimum size=0.7cm,inner sep=0pt] at (0,2) (4) {$x^\dagger_4$};
\node[draw,circle,minimum size=0.7cm,inner sep=0pt] at (0,1.2) (4B) {$\br{x}_4^\dagger$};
\node[draw,circle,minimum size=0.4cm,inner sep=0pt] at (-3.5,1.6) (4T) {$x^T_4$};
\node[draw,circle,minimum size=0.4cm,inner sep=0pt] at (3.5,1.6) (4F) {$x^F_4$};

\node[draw,circle,minimum size=0.7cm,inner sep=0pt] at (0,0) (5) {$x^\dagger_5$};
\node[draw,circle,minimum size=0.7cm,inner sep=0pt] at (0,-0.8) (5B) {$\br{x}_5^\dagger$};
\node[draw,circle,minimum size=0.4cm,inner sep=0pt] at (-3.5,-0.4) (5T) {$x^T_5$};
\node[draw,circle,minimum size=0.4cm,inner sep=0pt] at (3.5,-0.4) (5F) {$x^F_5$};

\draw[red!80!black]
  (T) -- (1T)  (T) -- (2T) (T) -- (3T)  (T) -- (4T) (T) -- (5T)
  (F) -- (1F)  (F) -- (2F) (F) -- (3F)  (F) -- (4F) (F) -- (5F);

  \draw[black]
  (1T) -- (1) (1T) -- (1B) (1F) -- (1) (1F) -- (1B)
  (2T) -- (2) (2T) -- (2B) (2F) -- (2) (2F) -- (2B)
  (3T) -- (3) (3T) -- (3B) (3F) -- (3) (3F) -- (3B)
  (4T) -- (4) (4T) -- (4B) (4F) -- (4) (4F) -- (4B)
  (5T) -- (5) (5T) -- (5B) (5F) -- (5) (5F) -- (5B);


\node[draw,circle,minimum size=0.7cm,inner sep=0pt] at (5,-2)  (p1){$x_1$};
\node[draw,circle,minimum size=0.7cm,inner sep=0pt] at (5,-2.8)  (p2){$\br{x}_2$};
\node[draw,circle,minimum size=0.7cm,inner sep=0pt] at (5,-3.6)  (p3){$x_3$};
\node[draw,circle,minimum size=0.4cm,inner sep=0pt] at (3,-2.4)  (pa){$a$};
\node[draw,circle,minimum size=0.4cm,inner sep=0pt] at (4,-3.2)  (pb){$b$};
\draw[black] (p2) -- (pb) (p3) -- (pb) (p1) -- (pa) (pa) -- (pb);

\node[draw,circle,minimum size=0.7cm,inner sep=0pt] at (1,-3)  (q1){$x_2$};
\node[draw,circle,minimum size=0.7cm,inner sep=0pt] at (1,-3.8)  (q2){$\br{x}_4$};
\node[draw,circle,minimum size=0.7cm,inner sep=0pt] at (1,-4.6)  (q3){$\br{x}_5$};
\node[draw,circle,minimum size=0.4cm,inner sep=0pt] at (-1,-3.4)  (qa){$a$};
\node[draw,circle,minimum size=0.4cm,inner sep=0pt] at (0,-4.2)  (qb){$b$};
\draw[black] (q2) -- (qb) (q3) -- (qb) (q1) -- (qa) (qa) -- (qb);

\node[draw,circle,minimum size=0.7cm,inner sep=0pt] at (-3,-4)  (r1){$\br{x}_1$};
\node[draw,circle,minimum size=0.7cm,inner sep=0pt] at (-3,-4.8)  (r2){$x_4$};
\node[draw,circle,minimum size=0.7cm,inner sep=0pt] at (-3,-5.6)  (r3){$x_5$};
\node[draw,circle,minimum size=0.4cm,inner sep=0pt] at (-5,-4.4)  (ra){$a$};
\node[draw,circle,minimum size=0.4cm,inner sep=0pt] at (-4,-5.2)  (rb){$b$};
\draw[black] (r2) -- (rb) (r3) -- (rb) (r1) -- (ra) (ra) -- (rb);

\draw[black]
(1) to[out=20,in=0, looseness=1.5] (p1)  (2B) to[out=40,in=0, looseness=2] (p2) (3) to[out=20,in=0, looseness=2]  (p3)
(2) to[out=-45,in=80, looseness=1.5] (q1)  (4B) to[out=-40,in=50, looseness=2] (q2) (5B) to[out=-40,in=20, looseness=2]  (q3)
(1B) to[out=-160,in=60, looseness=1.1] (r1)  (4) to[out=-130,in=50, looseness=2] (r2) (5) to[out=-130,in=20, looseness=1]  (r3);

\draw[red!80!black] (T) -- (pa) (T) -- (qa) (T) -- (ra);
\end{tikzpicture}
  \caption{$G_\phi$ for the 3CNF  formula $\phi = (x_1 \lor \br{x}_2 \lor x_3) \land (x_2 \lor \br{x}_4 \lor \br{x}_5) \land (\br{x}_1 \lor x_4\lor x_5)$.}\label{figure:hardness}
\end{figure*}

\section{Correlation Clustering on Complete Bipartite Graphs}
Let $(V = L \cup R, E)$ be a complete bipartite graph with $L$ and $R$ being the bipartition of the vertex set. In this section, we provide and analyze an algorithm for correlation clustering on complete graphs with an approximation guarantee of $5$ for minimizing the mistakes on one side of the bipartition (which without loss of generality will be $L$). The algorithm and analysis for complete bipartite graphs is very similar to the algorithm and analysis for complete graphs. At each step $t$ of our algorithm, we select a cluster center $w_t \in L$ and a cluster $C_t \subseteq (L \cup R)$ and remove it from the graph. This clustering step is repeated until all vertices in $L$ are part of some cluster. If there are any remaining vertices in $R$ which are unclustered, we put them in a single cluster.

Similar to the definition of $\Ball(w, \rho)$ in Section~\ref{sec:cor-clust-complete}, we define $\Ball_S(w, \rho) = \{u \in S : x_{uw} \leq \rho\}$. We select the cluster centers $w_t$ in step $t$ as follows. Let $V_t \subseteq V$ be the set of unclustered vertices at the start of step $t$. We redefine $L^S_t(w) = \sum_{u \in Ball_{V_t \cap S}(w, r)} r - x_{uw}$. We select $w_t$ as the vertex $w \in L$ that maximizes $L_t(w)$. We then select $Ball_{L \cup R}(w, 2r)$ as our cluster and repeat. A pseudocode for the above algorithm is provided in Figure~\ref{fig:Alg3}.

\begin{figure}
\notarxiv{\begin{center}}
\begin{algorithm}
\openLP
\smallskip

\noindent \textbf{Input: } Optimal solution $x$ to the linear program (P).\notarxiv{\\}

\par\noindent \textbf{Output: } Clustering $\calC$.

\medskip
\begin{enumerate}
\item Let  $V_0 = V$, $r = 1/5$, $t = 0$.
\item \textbf{while} ($V_t \cap L \neq \varnothing$)
  \begin{itemize}
	\item Find $w_t = \argmax\limits_{w \in L} L^R_t(w)$.
	\item Create a cluster $C_t = \Ball_{L \cup R}(w_t,2r)\cap V_t$.
	\item Set $V_{t+1} = V_t \setminus C_t$ and $t = t+1$.
\end{itemize}
\item Let $\calC_L = (C_0,\dots, C_{t-1})$.
\item \textbf{if} ($R \cap V_t \neq \emptyset$)
\begin{itemize}
	\item Let $C_R = R \cap V_t$.
\end{itemize}
\item Return  $\calC = \calC_L \cup \{C_R\}$.
\end{enumerate}
\notarxiv{\caption{Correlation Clustering on complete bipartite graphs}\label{alg:corelation-bipartite-complete}}
\end{algorithm}
\notarxiv{\end{center}}
\closeLP
\caption{Algorithm for Correlation Clustering on complete bipartite graphs.}\label{fig:Alg3}\label{alg:corelation-bipartite-complete}
\end{figure}

\subsection{Analysis}
In this section, we present an analysis of our algorithm.

\begin{theorem}
Algorithm 3 gives a $5$-approximation for Correlation Clustering on complete biparite graphs where disagreements are measured on only one side of the bipartition.
\end{theorem}

The proof of this theorem is almost identical to the proof of Theorem~\ref{thm:5-apx-main}. We define
$\lp{u,v}$, $\alg{u,v}$, $\prft{u,v}{t}$ for every edge $(u,v)$ and $\pft{u}$, $\prft{u}{t}$ for every vertex $u$
as in Section~\ref{sec:cor-clust-complete}. We then
show that for each vertex $u \in L$, we have $\pft{u} \geq 0$ and, therefore, the number of
disagreeing edges incident to $u$ is upper bounded by $5y(u)$:
$$ALG(u) = \smashoperator[r]{\sum_{v:(u,v) \in E}} \alg{u,v} \leq \frac{1}{r} \smashoperator[r]{\sum_{v:(u,v) \in E}} \lp{u,v} = 5y(u).$$
Thus, $\|ALG\|_q \leq 5 \|y\|_q$ for any $q\geq 1$. Consequently, the approximation ratio of the algorithm is at most $5$ for any norm $\ell_q$.

\begin{lemma}
For every $u\in L$, we have $\pft{u} \geq 0$.
\end{lemma}
As in Lemma~\ref{lem:pft}, we need to show that $\prft{u}{t}\geq 0$ for all $t$. Note that we only need
to consider $u\in V_t\cap L$ as $\prft{u}{t} = 0$ for $u\notin V_t$.

Consider a step $t$ of the algorithm and vertex $u\in V_t \cap L$. Let $w = w_t$ be the center of the cluster chosen at this step.
First, we show that since the diameter of the cluster $C_t$ is $4r$, for all negative edges $(u,v) \in E^-$ with $u,v \in C_t$, we can charge the cost of
disagreement to the edge itself, that is, $\prft{u,v}{t}$ is nonnegative for $(u,v)\in E^-$ (see Lemma~\ref{cl:neg-edge-profit-nenneg}). We then consider two cases:
$x_{uw}\in [0, r]\cup [3r,1]$ and $x_{uw}\in (r,3r]$.

The former case is fairly simple since disagreeing positive edges $(u,v)\in E^+$ (with  $x_{uw}\in [0, r]\cup [3r,1]$) have a ``large'' LP cost. In Lemma~\ref{lem:0r} and Lemma~\ref{lem:r1},
we prove that the cost of disagreement can be charged to the edge itself and hence $\prft{u}{t} \geq 0$.

We then consider the latter case. Similarly to Lemma~\ref{lem:pft}, we split the profit at step $t$ for vertices $u$ with $x_{uw} \in (r, 3r]$
into the profit they get from edges $(u,v)$ with $v$ in $\Ball_R(w,r)\cap V_t$ and from edges with $v$ in $V_t \setminus \Ball_R(w,r)$. That is,
\begin{multline*}
\prft{u}{t} =\\= \underbrace{\sum_{v\in \Ball_R(w,r) \cap V_t}\prft{u,v}{t}}_{P_{high}(u)} + \underbrace{\sum_{v\in V_t\setminus \Ball_R(w,r)}\prft{u,v}{t}}_{P_{low}(u)}.
\end{multline*}
Denote the first term by $P_{high}(u)$ and the second term by $P_{low}(u)$. We show that $P_{low}(u)\geq -L^R_t(u)$ (see Lemma~\ref{lem:PLow-LtRu} )
and $P_{high}\geq L^R_t(w) = \sum_{v \in \Ball_R{w, r} \cap V_t} r - x_{vw}$  (see Lemma~\ref{lem:PHigh-LtRu} ) and
 conclude that $\prft{u}{t} = P_{high}(u) + P_{low}(u)\geq L^R_t(w)-L^R_t(u)\geq 0$ since $L^R_t(w) = \max_{w'\in V_t} L^R_t(w') \geq L^R_t(u)$.

Consider $u$ such that $x_{uw} \in (r, 3r]$. First, we show that the profit we obtain from every edge $(u,v)$ with $v \in \Ball_R(w,r)$ is
at least $r - x_{vw}$, regardless of whether the edge is positive or negative.
\begin{claim}
If $x_{uw} \in (r,3r]$ and $v \in \Ball_R(w,r)\cap V_t$, then $\prft{u,v}{t}\geq r-x_{vw}$.
\end{claim}
\begin{proof}
First consider $u$ such that $x_{uw} \in (r, 2r]$. Note that $x_{uv} \geq x_{uw} - x_{vw} \geq r - x_{vw}$. Moreover, $x_{uv} \leq x_{uw} + x_{vw} \leq 2r + x_{vw}$. Thus, if $(u,v) \in E^+$, then $\prft{u,v}{t} \geq r - x_{vw}$. Otherwise, $\prft{u,v}{t} \geq (1 - 2r - x_{vw}) - r \geq 2r - x_{vw}$.

For $u \in (2r, 3r]$, note that $x_{uv} \geq x_{uw} - x_{vw} \geq 2r - x_{vw}$. Moreover, $x_{uv} \leq x_{uw} + x_{vw} \leq 3r + x_{vw}$.
Thus, if $(u,v) \in E^+$, then $\prft{u,v}{t} \geq (2r - x_{vw}) - r \geq r - x_{vw}$. Otherwise, $\prft{u,v}{t} \geq (1 - 3r - x_{vw}) \geq 2r - x_{vw}$.
\end{proof}

Using the above claim, we can sum up the profits from all vertices $v$ in $\Ball_R(w, r)$ and lower bound $P_{high}(u)$ as follows.
\begin{lemma}\label{lem:PHigh-LtRu}
If $x_{uw}\in (r,3r]$, then $P_{high}(u) \geq L^R_t(w)$.
\end{lemma}
\begin{proof}
By Claim , we have $\prft{u,v}{t}\geq r-x_{vw}$ for all $v\in R \cap V_t$. Thus,
\begin{align*}
P_{high}(u) &= \sum_{v\in \Ball_R(w,r)\cap V_t}\prft{u,v}{t}\\ &\geq \sum_{v\in \Ball_R(w,r)\cap V_t}r-x_{vw} = L^R_t(w).
\end{align*}
\end{proof}

We now lower bound $P_{low}(u)$. To this end. we estimate each term $\prft{u,v}{t}$ in the definition of $P_{low}$.
\begin{claim}\label{claim:lb-bp}
If $x_{uw} \in (r,3r]$ and $v \in V_t \setminus \Ball_R(w,r)$, then $\prft{u,v}{t}\geq \min(x_{uv} - r, 0)$.
\end{claim}
\begin{proof}
By Claim~\ref{cl:neg-edge-profit-nenneg}, if $(u,v)$ is a negative edge, then $\prft{u,v}{t} \geq 0$. The profit is $0$ if $x_{uv}\notin \Delta E_t$ (i.e., neither $u$ nor $v$ belong
to the new cluster). So let us assume that $(u,v)$ is a positive edge in $\Delta E_t$. Then, the profit obtained from $(u,v)$ is $x_{uv}$ if $(u,v)$ is in agreement
and $x_{uv} - r$ if $(u,v)$ is in disagreement. In any case, $\prft{u,v}{t} \geq x_{uv} - r \geq \min(x_{uv} - r, 0)$.
\end{proof}

Lemma~\ref{lem:PLow-LtRu} is an immediate corollary of Claim~\ref{claim:lb-bp}.
\begin{lemma}\label{lem:PLow-LtRu}
If $x_{uw}\in (r,3r]$, then $P_{low}(u) \geq -L^R_t(u)$.
\end{lemma}
\begin{proof}
By Claim~\ref{claim:prof-uv-lower-bound}, we have $\prft{u,v}{t}\geq \min(x_{uv} - r,0)$ for all $v\in V_t$. Thus,
\begin{align*}
P_{low}(u) &= \sum_{v\in V_t\setminus \Ball_R(w,r)}\prft{u,v}{t}\\
&\geq  \sum_{v\in V_t\setminus \Ball_R(w,r)} \min(x_{uv} - r,0)\\
&\overset{a}{\geq} \;\;\;\;\;\sum_{v\in V_t} \min(x_{uv} - r,0) \\
&\overset{b}{=} \sum_{v\in \Ball_R(u,r) \cap V_t} x_{uv} - r \\
&= - L_t^R(u).
\end{align*}
Here we used that (a) all terms $\min(x_{uv} - r,0)$ are nonpositive, and (b) $\min(x_{uv} - r, 0) = 0$ if $v\notin \Ball(u,r)$.
\end{proof}

\end{document}